\documentclass[runningheads]{llncs}
\usepackage[T1]{fontenc}
\usepackage[margin=3.5cm]{geometry}
\usepackage{amsmath, amsfonts, amssymb, mathdots, mathabx}
\usepackage{graphicx, svg, subcaption, tikz}
\usetikzlibrary{arrows.meta, decorations.markings, quotes, shapes.geometric}

\usepackage[nottoc]{tocbibind}  

\usepackage{xcolor}
\definecolor{ForestGreen}{RGB}{34,139,34}  
\usepackage[colorlinks, citecolor=blue, linkcolor=ForestGreen, urlcolor=blue, pdfborder={0 0 0}]{hyperref}

\usepackage{algorithm, algorithmic}
\usepackage{lmodern, setspace, enumerate, enumitem, multirow, rotating, makeidx, comment, soul, diagbox, makecell, caption, verbatim}
 \usepackage{mdframed, ragged2e}

\tikzset{currarrow/.style={-Stealth}}

\newcommand{\naturals}{\mathbb{N}}
\newcommand{\integers}{\mathbb{Z}}
\newcommand{\reals}{\mathbb{R}}

\newcommand{\coss}[1]{\cos\mkern-2mu\left(#1\right)}
\newcommand{\arccoss}[1]{\arccos\mkern-2mu\left(#1\right)}
\newcommand{\sinn}[1]{\sin\mkern-2mu\left(#1\right)}
\newcommand{\tann}[1]{\tan\mkern-2mu\left(#1\right)}

\renewcommand{\d}[1]{\operatorname{d}\mkern-2mu{#1}}
\newcommand{\dd}{\mathrm{d}}

\newcommand{\shoreline}[1]{$\mathcal S_{#1}$}
\newcommand{\pshoreline}[1]{$\mathcal P(#1)$}
\newcommand{\tradeoff}[1]{$\mathcal T(#1)$}
\newcommand{\ashoreline}[3]{$\mathcal{A}\mkern-4mu\left(#1,#2,#3\right)$}
\newcommand{\inspect}[1]{\mathcal I(#1)}

\makeatletter
\renewcommand\subsubsection{\@startsection{subsubsection}{3}{\z@}%
  {-3.25ex\@plus -1ex \@minus -.2ex}%
  {1.5ex \@plus .2ex}%
  {\normalfont\normalsize\bfseries}}%
\makeatother

\newcommand{\ignore}[1]{}

\newtheorem{aalgorithm}[theorem]{Algorithm}

\begin{document}

\title{
Disk and Partial Disk Inspection: \\
Worst- to Average-Case and Pareto Upper Bounds
\thanks{
This is the full version of the paper~\cite{conley2024multi-sirocco}, which is included in the proceedings of the 32nd International Colloquium on Structural Information and Communications Complexity (SIROCCO 2025), held on June 2–4, 2025, in Delphi, Greece.
}
}

\author{
James Conley
\and 
Konstantinos Georgiou
\thanks{Research supported in part by NSERC Discovery grant
and by the Toronto Metropolitan University Faculty of Science Dean’s Research Fund.
}
}
\authorrunning{J. Conley and K. Georgiou}

\institute{Department of Mathematics, 
Toronto Metropolitan University, 
Toronto, ON, Canada
\email{j1conley,konstantinos@torontomu.ca}
}

\maketitle

\begin{abstract}
We consider $n$ unit-speed mobile agents initially positioned at the center of a unit disk, tasked with inspecting all points on the disk’s perimeter. 
A perimeter point is considered covered if an agent positioned outside the disk’s interior has unobstructed visibility of it, treating the disk itself as an obstacle. 
For $n=1$, this problem is referred to as the shoreline problem with a known distance and was originally proposed as a more tractable variant of Bellman's famous lost-in-the-forest problem~\cite{bellman1956minimization}. 
Isbell~\cite{isbell1957optimal} in 1957 derived an optimal trajectory that minimizes the worst-case inspection time for that problem, while Gluss~\cite{gluss1961alternative} in 1961 proposed some heuristics for the average-case inspection problem. 
For $n\geq 2$ agents, worst-case optimal trajectories were shown in~\cite{AcharjeeGKS19,dobrev2020improved}.

Our contributions are threefold. First, and as a warm-up, we extend Isbell’s findings by deriving worst-case optimal trajectories addressing the \emph{partial inspection} of a section of the disk, hence deriving an alternative proof of optimality for inspecting the disk with $n \geq 2$ agents. Second, we analyze the average-case inspection time, assuming a uniform distribution of perimeter points (equivalent to randomized inspection algorithms). Using spatial discretization and Nonlinear Programming (NLP), we propose feasible solutions to the continuous problem and evaluate their effectiveness compared to NLP solutions, improving upon the heuristics reported in~\cite{gluss1961alternative}. 
Third, we establish Pareto-optimal bounds for the multi-objective problem of jointly minimizing the worst-case and average-case inspection times.

\vspace{0.5cm}
\noindent
{\bf Mobile Agents, Search, Disk, Shoreline Problem}
\end{abstract}

\newpage
\setcounter{tocdepth}{3}
\tableofcontents
\newpage

\section{Introduction}
In 1955, Bellman~\cite{bellman1956minimization} posed a seminal  question in search theory: given a random point $P$ within a region $R$, what trajectory minimizes the expected time to reach the boundary, and what is the maximum time required? This question laid the groundwork for search theory, inspiring subsequent research into optimal trajectories under uncertainty. Bellman also proposed studying specific cases, such as when $R$ represents the area between two parallel lines separated by a known distance or when $R$ is a semi-infinite plane with a known distance to the boundary.
This problem also appears in Bellman's classic work on dynamic programming~\cite{bellman1958dynamic} (Problem 6, page 133), now recognized as the ``lost-in-a-forest'' problem. Over the decades, this question has inspired various modifications, with only a few specific cases fully resolved and many still open (see Section~\ref{sec: related work} for additional references).
This question is a central reference in search theory~\cite{berzsenyi1995lost} and has even appeared in the USSR Olympiad Problem Book~\cite{ussr_olympiad_russian}. Notably, the variant with two parallel lines is among the few examples that have been fully solved (see Problem 7, page 133, in~\cite{bellman1958dynamic}).

Bellman~\cite{bellman1956minimization} introduced a notable variation of the problem as a simpler alternative to the lost-in-a-forest problem. This worst-case variation, solved by Isbell~\cite{isbell1957optimal} in 1957, is known as the \emph{shoreline problem with known distance}. The problem seeks a minimum-length curve in $ \mathbb{R}^2 $ whose convex hull contains a disk of radius $1$ centered at one endpoint. Alternatively, in modern mobile agent terms, it involves a unit-speed mobile agent tasked with finding a hidden infinite (shore)line at a known distance of $1$, though its location is unknown. Isbell’s solution provides the agent’s optimal trajectory, minimizing the worst-case time to reach the shoreline.
Motivated by the optimal results in~\cite{isbell1957optimal}, and in particular by properties of the underlying optimal trajectories for the worst-case problem, Gluss~\cite{gluss1961alternative} proposed heuristic solutions to the average-case problem, yielding the first (and to our knowledge, the only) upper bounds.

The shoreline can also be interpreted as the visibility boundary for an agent tasked with inspecting (or patrolling) the perimeter of a unit disk, starting from its center. This application could represent inspecting the perimeter of a circular building for defects or intruders, where visibility of the exterior is possible only from outside the building. Specifically, a point on the disk is considered visible (inspected) if no convex combination of the agent’s position and the point lies within the disk’s interior, treating the disk as an obstacle to visibility.
We refer to this task as the \emph{1-agent inspection problem}, which Isbell resolved optimally by minimizing the worst-case inspection cost.

To the best of our knowledge, we are the first to study the \emph{multi-agent} variant of this problem in both worst- and average-case settings. For the worst-case setting, we derive optimal results for the partial inspection problem, which yield alternative proofs of optimality when inspection is performed by $n$ agents. For the average-case setting, which is the main novelty of this work, we introduce a systematic methodology based on discretization and nonlinear programming, and we establish upper bounds that improve on the heuristics of~\cite{gluss1961alternative} both in value and in method. This framework also applies to the multi-objective setting, where we provide Pareto upper bounds that balance worst-case and average-case costs.

\subsection{Related Work}
\label{sec: related work}

The study of search problems began in the 1950s and 1960s~\cite{beck1964linear,bellman1963optimal} and has since become a distinct field within operations research and theoretical computer science. Foundational surveys and texts~\cite{ahlswede1987search,alpern2013search,alpern2006theory,CGK19search,gal2010search} document its development. Early research focused on single-agent line search~\cite{baeza1988searching,kleinberg1994line}, while advances in robotics spurred interest in multi-agent variants~\cite{CzyzowiczGGKMP14,pattanayak2018evacuating}.

In search problems, agents aim to locate hidden targets within a domain. 
Linear search problems, such as the classic cow-path problem, are fundamental and have been extensively studied for single~\cite{baezayates1993searching} and multiple agents~\cite{ChrobakGGM15}. Recent work introduces variations minimizing weighted completion times~\cite{GLweightedLine2023} and adapts searches to complex domains such as rays~\cite{BrandtFRW20}, terrains~\cite{CzyzowiczKKNOS17}, and graphs~\cite{AngelopoulosDL19}.

Two-dimensional search introduces added complexity, with early results for polygons~\cite{FeketeGK10} and disks~\cite{CzyzowiczGGKMP14}, and expanded to include shapes like triangles, squares, $\ell_p$ disks, and regular polygons~\cite{BagheriNO19,CzyzowiczKKNOS15,georgiou2022triangle,GLLKllp2023,georgiou2025multi}. Multi-agent models in these contexts address coordination, speed~\cite{BampasCGIKKP19}, and fault tolerance~\cite{behrouz2023byzantine,BGMP2022pfaulty,CzyzowiczGGKKRW17}. The only related work addressing worst-case and average-case trade-offs in mobile agent search is in~\cite{chuangpishit2020multi}, examining single-agent search and evacuation in a unit disk.

The two-dimensional search problem most relevant to our work is Bellman's shoreline problem~\cite{bellman1958dynamic}, which was first solved optimally in the worst-case for a single agent by Isbell~\cite{isbell1957optimal}, with the target shoreline at a known distance from the agent's starting point. Searching for a shoreline at a known distance with $n \geq 4$ agents was solved optimally in the worst-case sense in~\cite{AcharjeeGKS19}, while the cases of $n=2,3$ were treated in~\cite{dobrev2020improved}. To the best of our knowledge, the only upper bounds known for the average-case problem are due to the heuristics in~\cite{gluss1961alternative}, where it was conjectured that optimal trajectories share structural properties with worst-case solutions. In the present work, we systematize the search for efficient trajectories and obtain inspection strategies that not only improve the previously reported upper bounds for the single-agent case, but also, interestingly, diverge from the conjectured structural properties.

The more general lost-in-a-forest problem has been surveyed in~\cite{berzsenyi1995lost,finch2004lost}. Recent progress on the area include that for convex polygons in~\cite{gibbs2016bellman} and well as for general shapes using competitive analysis in~\cite{kubel2021approximation}. 
Searching for a (shore)line without knowledge of its distance was introduced for a single robot in~\cite{baeza1988searching}, with extensions for $n \geq 2$ robots explored in~\cite{baeza1997searching,baeza1995parallel,baezayates1993searching,jez2009two}. For $n=1$, a logarithmic spiral achieves the best known competitive ratio of $13.81$~\cite{baeza1988searching,finch2005searching}, with lower bounds of $6.3972$ unconditionally~\cite{baeza1995parallel} and $12.5385$ for cyclic trajectories~\cite{langetepe2012searching}. 
For multiple robots, a double logarithmic spiral yields a competitive ratio of $5.2644$ for $n=2$~\cite{baeza1995parallel}, and for $n \geq 3$, robots moving along rays achieve competitive ratio at most $1/\cos(\pi/n)$~\cite{baeza1995parallel}, with matching lower bounds provided for $n \geq 4$ and $n=3$ in~\cite{AcharjeeGKS19} and~\cite{dobrev2020improved}, respectively. Variants where partial information is available yield varying competitive ratios based on distance, slope, or orientation knowledge~\cite{baeza1988searching,isbell1957optimal,gluss1961minimax,jez2009two}.

Beyond shoreline search, strategies have been developed for circles~\cite{gluss1961minimax} and other planar regions. Langetepe~\cite{langetepe2010optimality} proved that the spiral search is optimal for single-agent point search in two dimensions, while multi-agent and grid-based searches with memory constraints have been examined in~\cite{emek2015many,Emekicalp2014,fricke2016distributed,LangnerKUW15}. Recent studies address point search in planar and geometric terrains~\cite{bouchard2018deterministic,pelc2018reaching}, with multi-speed agents~\cite{georgiou2025spirals} and explore trade-offs between search cost and information~\cite{pelc2018information,pelc2019cost}.

\subsection{Roadmap}

Section~\ref{sec: Inspection Problems and New Contributions} presents the necessary preliminaries and result statements. Specifically, Section~\ref{sec: definitions notation and auxiliary problems} formally introduces the problems studied and defines the notation, while Section~\ref{sec: Formal Description of New Results} summarizes all results with appropriate quantification, and in Section~\ref{sec: comparison} we provide a comprehensive comparison to previously known results. 

The worst-case analysis begins in Section~\ref{sec: worst case inspection for shoreline n agents}, with key inspection-related observations made and proved in Section~\ref{sec: Some Preliminary Observations}. Section~\ref{sec: solve inspection by partial inspection} then addresses the partial inspection problem, followed by the $n$-agent inspection problem for the worst case.
The average-case problem is examined in Section~\ref{sec: average case inspection for shoreline n agents}. Specifically, Section~\ref{sec: discrete average} presents a discretized algorithm for the problem and quantifies its performance. This foundation allows us to address the partial continuous average inspection problem in Section~\ref{sec: solution to avg partial and shoreline-n}.
Finally, Section~\ref{sec: tradeoffs} explores the multi-objective problem, and Section~\ref{sec: conclusions} concludes the paper with a discussion of the established results.

\section{Inspection Problems and New Contributions}
\label{sec: Inspection Problems and New Contributions}

\subsection{Formal Definition \& Some Auxiliary Problems}
\label{sec: definitions notation and auxiliary problems}

We formally define the main problems studied below, along with auxiliary problems whose solutions aid in deriving key results. All definitions are provided in this section for convenience.

\paragraph{$n$-agent Disk Inspection Problem, \shoreline{n}:}
In this problem, $n$ unit-speed agents (searchers) start from the center of a unit-radius disk. The agents can move freely within the plane of the disk, changing directions without added time costs. For a given trajectory, a point $P$ on the disk's perimeter is said to be inspected if an agent passes through a point $A$ in the plane such that no convex combination of $A$ and $P$ lies within the disk. The first time $P$ is inspected by an agent is its \emph{inspection time}, denoted by $\inspect{P}$. Feasible solutions to this problem are agents' trajectories that inspect every point on the disk's perimeter. In \shoreline{n}, the focus is on minimizing either the worst-case inspection time, $\sup_P \inspect{P}$, or the average-case inspection time, $\mathbb{E}_P \inspect{P}$, where the expectation is taken over a uniform distribution of $P$.

The problem of minimizing $\sup_P \inspect{P}$ for \shoreline{1} was originally introduced and optimally solved in~\cite{isbell1957optimal}. Also known as the \textit{shoreline problem with known distance}, it seeks the shortest trajectory of a unit-speed agent to reach a hidden line located at a distance of 1 from the agent. Equivalently, the problem asks for the shortest continuous 1-dimensional curve, with one endpoint at the disk's center, whose convex hull encloses the disk. Minimizing $\sup_P \inspect{P}$ for \shoreline{n} extends this description with $n$ such curves.

For the problem of minimizing $\mathbb{E}_P \inspect{P}$ in \shoreline{n}, the relevant trajectories correspond to deterministic algorithms, evaluated by their average-case performance. Alternatively, these algorithms can be viewed as randomized, where the circle is uniformly rotated before execution. Under this randomized interpretation, the expected inspection time $\inspect{P}$ is identical for all points $P$ and matches the expected performance of the deterministic algorithm under a uniformly random choice of $P$. Thus, minimizing $\mathbb{E}_P \inspect{P}$ can also be framed as minimizing the worst-case (expected) performance of a randomized algorithm.

\paragraph{Single-Agent Partial Disk Inspection Problem, \pshoreline{c}:}
Our approach to minimizing $\sup_P \inspect{P}$ for \shoreline{n} is facilitated by examining the single-agent problem \pshoreline{c}, which involves inspecting an arc of length $c \in [0, 2\pi]$ on the disk’s perimeter. Noting that in an optimal solution for \shoreline{n}, all agents should conclude their search simultaneously, each agent should inspect an arc of length $2\pi/n$. Thus, \shoreline{n} is equivalent to \pshoreline{2\pi/n}.

\paragraph{Single-Agent Auxiliary Disk Inspection Problem, \ashoreline{k}{\theta}{c}:}
This problem is defined for parameters $k \in \mathbb{N} \cup \{\infty\}$, $\theta < \pi/2$, and $c \geq 2\theta$. It is inspired by optimal (worst-case performance) algorithms for \shoreline{n}. Here, a single-agent partial inspection problem is introduced where the agent’s starting position is not at the origin but at a carefully selected point outside the disk’s interior. This setup enables deriving results also for minimizing $\mathbb{E}_P \inspect{P}$ for \shoreline{n}.

The parameter $k$ designates whether the inspection domain is discrete or continuous, with $k$ specifying the number of discrete inspection points if applicable. The parameter $c$ represents the range of points to be inspected, and $\theta$ specifies the agent's initial position. For simplicity, we consider a unit-radius disk centered at the origin in the 2D Cartesian plane, denoted $O$, with its perimeter parameterized as 
$$\mathcal{P}_t := (\coss{t}, \sinn{t}),$$
where $t \in [0, 2\pi].$

In \ashoreline{k}{\theta}{c}, the agent starts at $(1, \tan{\theta})$. When $k = \infty$, the objective is to inspect points $P_\phi$ with $\phi \in [2\theta, c]$, noting that points $P_\phi$ in $[0, 2\theta]$ are already inspected due to the agent's initial placement, as we will prove later. When $k \in \mathbb{N}$, the goal is to inspect discrete points $P_i := \mathcal{P}_{\phi_i}$, where
$$
\phi_i := 2\theta + \frac{c - 2\theta}{k} i, \quad i = 0, \ldots, k.
$$
Here, $\phi_i$ implicitly depends on the fixed values of $k$, $\theta$, and $c$. As with \shoreline{n}, \ashoreline{k}{\theta}{c} may aim to minimize either the worst-case or average-case inspection time. Later, we will connect and quantify solutions for \ashoreline{\infty}{\theta}{c} and \ashoreline{k}{\theta}{c}, ultimately relating them to \shoreline{n} for carefully selected values of $k$, $\theta$, and $c$ based on $n$.

\paragraph{Single-Agent Trade-off Problem \tradeoff{\lambda}:}
With our results for \shoreline{n} established, we consider the multi-objective optimization problem of minimizing both worst-case and average-case inspection times with a single agent. Alternatively, we approach bounding the Pareto-optimal set for these objectives by minimizing
$$
\lambda \cdot \sup_P \inspect{P} + (1-\lambda) \cdot \mathbb{E}_P \inspect{P},
$$
where $\lambda \in [0,1]$. The trade-off problem \tradeoff{\lambda} addresses real-life applications in which minimizing average performance is prioritized, subject to a manageable worst-case bound, such as cases with energy constraints that limit the operational duration of mobile agents.

\subsection{Main Results Made Formal \& Motivation}
\label{sec: Formal Description of New Results}

In this section we outline our main contributions, while in the following section, we provide a detailed comparison to known results. 
The motivating starting point for the present work is the optimal worst-case solution to \shoreline{1}~\cite{isbell1957optimal} (or equivalently to $\mathcal P(2\pi)$) which we generalize with the following result. 

\begin{theorem}
\label{thm: optimal worst case partial inspection sol}
The optimal worst case inspection cost to \pshoreline{c} is 
$$
\begin{cases}
    \frac{1}{\coss{c/2}}, & \text{if } 0 \leq c \leq \frac{2\pi}{3}, \\
    1 - 2\coss{c}, & \text{if } \frac{2\pi}{3} < c \leq \frac{5\pi}{6}, \\
    1 + \sqrt{3} + c - \frac{5\pi}{6}, & \text{if } \frac{5\pi}{6} < c \leq 2\pi.
\end{cases}
$$
\end{theorem}

As a corollary, (and setting $c=2\pi/n$), we obtain an alternative proof of the following known optimal results. 

\begin{theorem}
\label{thm: general wrs optimal}
The optimal worst case cost to \shoreline{n} ia
$$
\begin{cases}
1+\sqrt{3}+2\pi/n -5\pi/6,		& \textrm{ if }~ n=1,2, \\
1/\coss{\pi/n},					& \textrm{ if }~ n\geq 3. \\
\end{cases}
$$
\end{theorem}

Our next main result concerns upper bounds for the average-case inspection time of \shoreline{n}, which for $n=1$ improves upon the previously reported bounds in~\cite{gluss1961alternative}. Notably, the results in~\cite{gluss1961alternative} contain a numerical error, which we identify and correct in the next section, where we also provide a detailed comparison between their heuristic approach and ours.

\begin{theorem}
\label{thm: general avg upper bound}
\shoreline{n} can be solved in average-case inspection time as outlined in the next table. 
\begin{table}[h!]
\centering
\begin{tabular}{|c|c|c|c|}
\hline
$n$                   & $1$                  & $2$                  & $\geq 3$ \\ \hline
average cost  & $3.5509015$          & $1.7946051$          & $\frac{n}{2\pi} \log \left( \frac{1+\sin{\pi/n}}{1-\sin{\pi/n}} \right)$ \\ \hline
\end{tabular}
\caption{Upper bounds to the average inspection cost of \shoreline{n}.
Indicatively, we also get performance equal to $\frac{3 \log \left(4 \sqrt{3}+7\right)}{2 \pi }\approx 1.2576$ and $\frac{2 \log \left(2 \sqrt{2}+3\right)}{\pi }\approx 1.1222$, for 
$n=3,4$, respectively, while the values continue to drop and asymptotically tend to $1+\Theta(1/n^2)$. 
}
\end{table}
\ignore{
$n$   			average cost 								$\theta$
$1$			$3.5509015$							$0.5910554$	
$2$			$1.7946051$							$0.8054878$
$\geq 3$		$\frac{n}{2\pi} \log \left( \frac{1+\sinn{\pi/n}}{1-\sinn{\pi/n}} \right)$		$\pi/n$
\begin{table}[h!]
\centering
\begin{tabular}{|c|l|l|}
\hline
$n$            & \multicolumn{1}{c|}{average cost}                                                     & $\theta$    \\ \hline
$1$            & $3.5509015$                                                                            & $0.5910554$ \\ \hline
$2$            & $1.7946051$                                                                            & $0.8054878$ \\ \hline
$\geq 3$       & $\frac{n}{2\pi} \log \left( \frac{1+\sinn{\pi/n}}{1-\sinn{\pi/n}} \right)$                & $\pi/n$     \\ \hline
\end{tabular}
\caption{Cost and Angle Values}
\end{table}
}
\end{theorem}
As before, our results are based on designing efficient average-cost inspection algorithms for the general problem \pshoreline{c}, $c \in [0, 2\pi]$, which may be of independent interest. The reader is referred to Figure~\ref{fig: avg case trajectories n=1,2} for the derived trajectories minimizing the average-cost inspection time for $n=1,2$, corresponding to solutions for \pshoreline{c} with $c=\pi$ and $c=2\pi$. Figure~\ref{fig: trajectories 5,6,7pOver4} on page~\pageref{fig: trajectories 5,6,7pOver4} displays the optimal trajectories derived for \pshoreline{c} for additional values of $c$.
\begin{figure}[h!]
    \centering
    \begin{subfigure}[t]{0.47\textwidth}
        \centering
        \includegraphics[width=5cm]{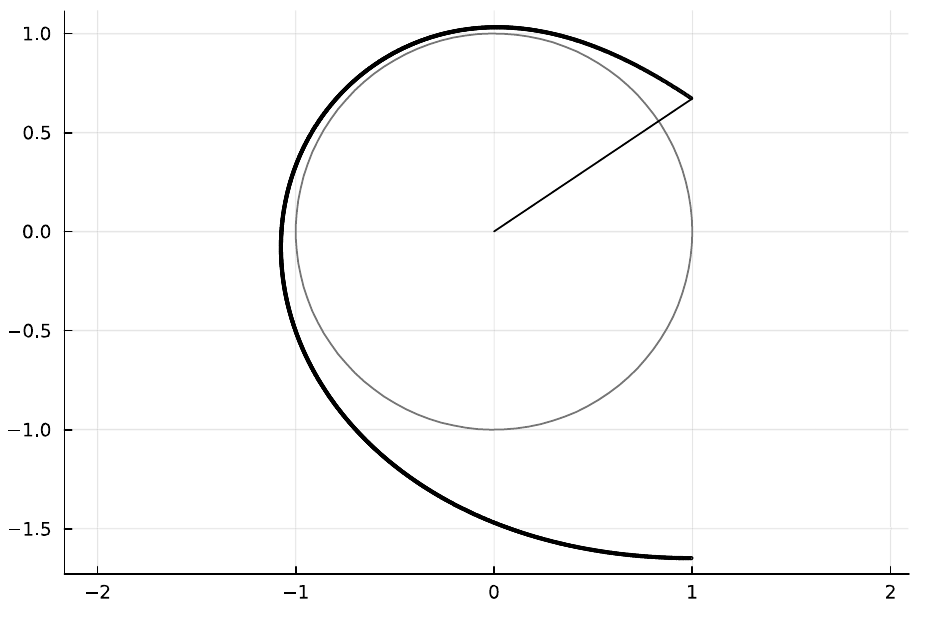}
        \caption{The trajectory proving the average inspection upper bound to \shoreline{1}.}
        \label{fig: avg trajectory n=1}
    \end{subfigure}
    \hfill
    \begin{subfigure}[t]{0.47\textwidth}
        \centering
        \includegraphics[width=4.8cm]{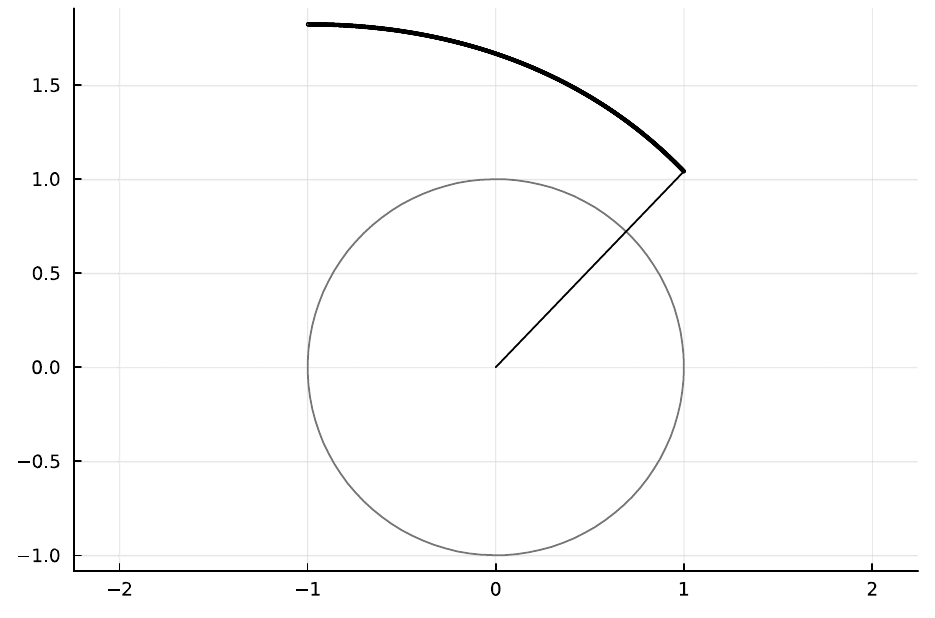}
        \caption{The trajectory proving the average inspection upper bound to \shoreline{2}.}
        \label{fig: avg trajectory n=2}
    \end{subfigure}
    \caption{Agent's trajectories of Theorem~\ref{thm: general avg upper bound} for $n=1,2$.}
    \label{fig: avg case trajectories n=1,2}
\end{figure}

Notably, the naive inspection algorithm for \shoreline{n}, where each agent inspects an arc of length $2\pi/n$ along the disk, yields an average inspection cost of $1 + \pi/n$, which is approximately 4.14159 and 2.5708 for $n=1$ and $n=2$, respectively. This improvement is substantial for all values of $n$, and asymptotically, our results show that the average inspection cost converges to $1 + \frac{\pi^2}{6n^2} + o\left( \frac{1}{n^3}\right)$.

In Section~\ref{sec: tradeoffs}, we analyze the trade-offs between worst-case and average-case inspection costs for \shoreline{1}, focusing on minimizing the multi-objective function $\{\sup_P \inspect{P}, \mathbb{E}_P \inspect{P}\}$. This is achieved by minimizing
$
\lambda \cdot \sup_P \inspect{P} + (1-\lambda) \cdot \mathbb{E}_P \inspect{P},
$
for $\lambda \in [0,1]$. Theorem~\ref{thm: general wrs optimal} (originally due to~\cite{isbell1957optimal} for \shoreline{1}) provides the worst-case optimal inspection time for \shoreline{1} as $1+\sqrt{3}+\frac{7 \pi }{6} \approx 6.39724$, while the same trajectory has an average inspection cost of 
$$\frac{5}{2 \sqrt{3}}+\frac{91 \pi }{144}+\frac{\log(6)}{2 \pi } \approx 3.71386,$$
as shown in Lemma~\ref{lem: isbell performance}.

At the other end, Theorem~\ref{thm: general avg upper bound} gives an upper bound on the average inspection cost at 3.5509015, corresponding to a trajectory length of approximately 6.8673829. We are thus motivated to establish upper bounds on average-case performance when the trajectory length varies between 6.39724 and 6.8673829.
To achieve this, we first determine the Pareto curve for the Isbell-type trajectories used in the proof of Theorem~\ref{thm: original Isbell}, parameterized by $\theta$. For suitable values of $\theta$, we compute in Lemma~\ref{lem: isbell performance} both the worst-case and average-case costs for this family of algorithms, which serve as benchmarks for our subsequent results.

Finally, to improve the upper bounds for the multi-objective inspection problem, we extend our approach to minimize average-case inspection costs. This is accomplished using a discrete trajectory optimized through a Non-Linear Program with carefully selected parameters. The bounds for the Pareto curve of the multi-objective inspection problem are summarized in the following theorem.

\begin{theorem}
\label{thm: pareto bounds}
Single Agent trade-off problem \tradeoff{\lambda}, where $\lambda \in [0,1]$ admits inspection trajectories with worst-case cost and average-case cost to \shoreline{1} that are depicted in Figure~\ref{fig: all tradeoff}.
\begin{figure}[h!]
    \centering
    \includegraphics[width=0.6\textwidth]{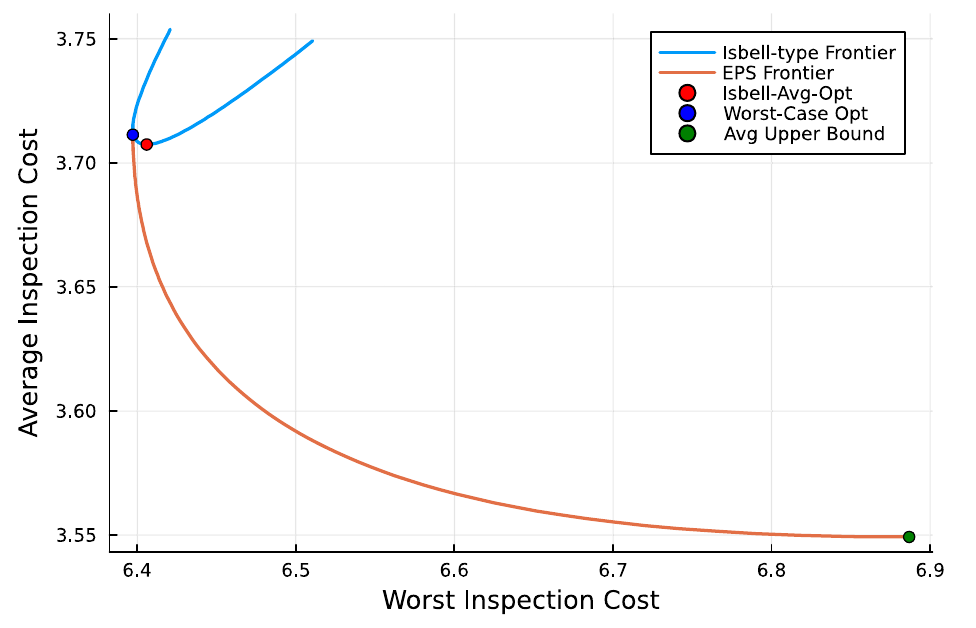}
    \caption{
        Our Pareto upper bounds are depicted in orange and are achieved by the Extended-Poly-Segment (EPS) trajectory we define later. The blue line depicts the Pareto Frontier of the Isbell-type trajectories, parameterized by the initial deployment angle $\theta$. The provable worst-case optimal solution and its corresponding average inspection performance are shown as the blue Pareto point. The red Pareto point corresponds to the performance of the Isbell-type trajectory that is average cost optimal. Finally, the green Pareto point is the performance of the EPS-type trajectory with the smallest average inspection cost.
    }
    \label{fig: all tradeoff}
\end{figure}
\end{theorem}

\subsection{Comparison to Known Results and the Work of Gluss~\cite{gluss1961alternative} (1961)}
\label{sec: comparison}

In this section, we place our contributions from Section~\ref{sec: Formal Description of New Results} in the context of existing results. For the worst-case Disk Inspection problem, we review the known optimal trajectories. For the average-case problem, we examine in greater detail the early work of Gluss~\cite{gluss1961alternative}, identify a numerical error in his reported bounds, and contrast his heuristic approach with our systematic method.

\begin{figure}[h!]
    \centering
    \begin{subfigure}[t]{0.32\textwidth}
        \centering
        \includegraphics[width=4.6cm]{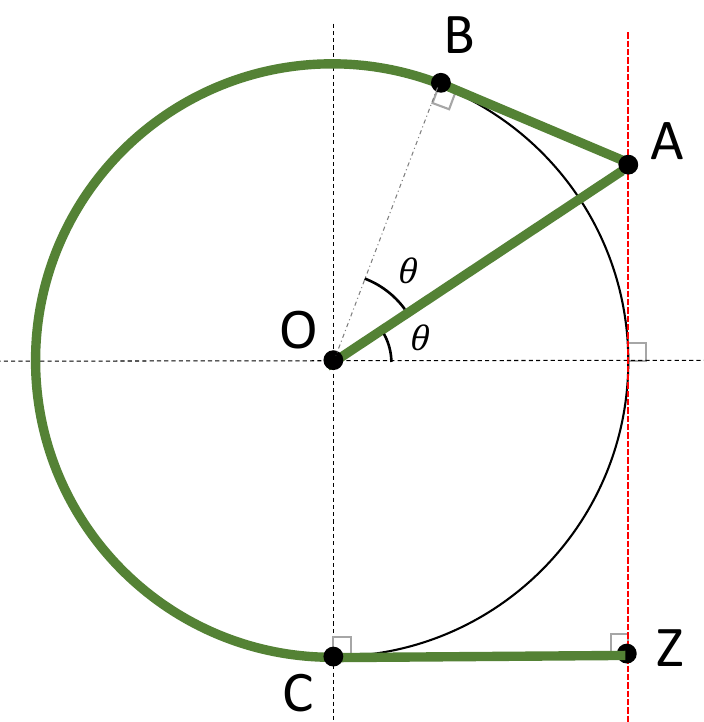}
        \caption{Isbell’s optimal worst-case trajectory~\cite{isbell1957optimal} for \shoreline{1}.}
        \label{fig: isbel}
    \end{subfigure}
    \hfill
    \begin{subfigure}[t]{0.32\textwidth}
        \centering
        \includegraphics[width=4.6cm]{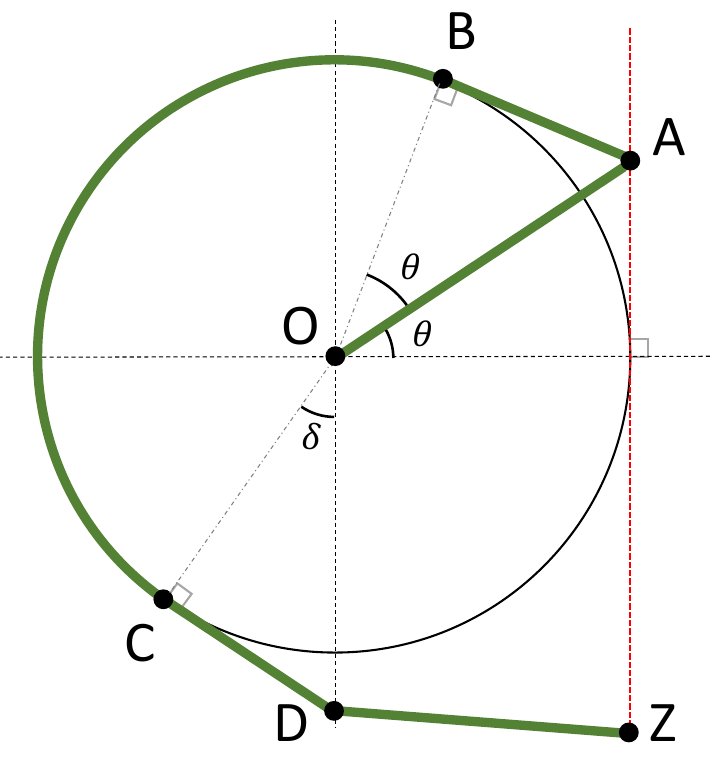}
        \caption{First heuristic trajectory of Gluss~\cite{gluss1961alternative}.}  
\label{fig: gluss1}
    \end{subfigure}
    \hfill
    \begin{subfigure}[t]{0.32\textwidth}
        \centering
        \includegraphics[width=4.6cm]{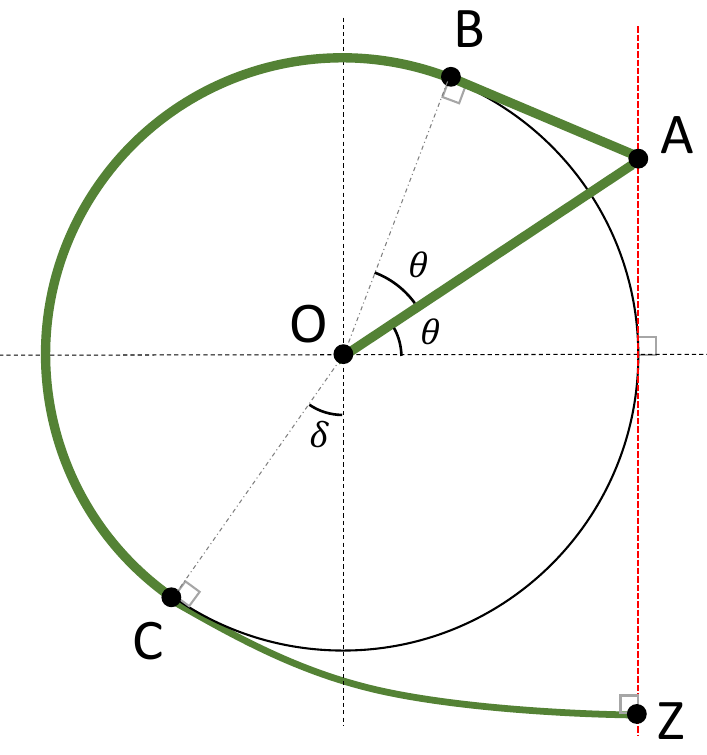}
        \caption{Second heuristic trajectory of Gluss~\cite{gluss1961alternative}.}  
\label{fig: gluss2}
    \end{subfigure}
    \caption{Previously known trajectories for Disk Inspection \shoreline{1}.}
    \label{fig: previous results}
\end{figure}

\subsubsection{Improvements to the Worst-Case Disk Inspection}

The first result on Disk Inspection, originally phrased as searching for a shoreline at known distance, is due to Isbell~\cite{isbell1957optimal}, and reads as follows. 

\begin{theorem}
\label{thm: original Isbell}
The optimal worst-case cost to \shoreline{1} is $1+\sqrt{3}+\tfrac{7 \pi }{6} \approx 6.39724$. 
\end{theorem}

The optimal trajectory is illustrated in Figure~\ref{fig: isbel}. Consider the unit disk centered at the origin $O$ in the Cartesian plane. The agent moves first to a point $A$ on the line $x=1$ in the first quadrant, chosen so that $OA$ makes an angle $\theta \in [0,\pi/2)$ with the horizontal axis. From $A$ the trajectory continues along the tangent $AB$ to the disk at point $B$ (counterclockwise, with $A=B$ if $\theta=0$). The agent then traverses the perimeter from $B$ to $C=\mathcal P_{3\pi/2}$, and finally moves along the projection from $C$ to $D$ on the line $x=1$ in the fourth quadrant. The cost of this trajectory is a function of $\theta$, optimized at $\theta=\pi/6$, yielding the bound in Theorem~\ref{thm: original Isbell}.

This trajectory is notable in that it follows a long arc of the disk boundary (segment $BC$). Related constructions yield optimal worst-case trajectories for \shoreline{2} and \shoreline{3}, as shown by Dobrev et al.~\cite{dobrev2020improved}, with the final tangent chosen at $\mathcal P_{\pi}$ (line $x=-1$, see Figure~\ref{fig: avg trajectory n=2}) and $\mathcal P_{2\pi/3}$, respectively. For $n \geq 4$, optimal trajectories were established in~\cite{AcharjeeGKS19}. For $n\geq 3$, the optimal solution consists of a single line segment from the origin of length $1/\coss{\theta}$, inspecting an arc of length $2\theta$ (see segment $OA$ in Figure~\ref{fig: initialpoint}). Although the arguments in~\cite{AcharjeeGKS19,dobrev2020improved} were originally given in the unknown-distance setting, the lower-bound proofs apply also when the distance is known, effectively showing that the adversary gains no advantage from that information in these cases.

All these results for \shoreline{n} are summarized in Theorem~\ref{thm: general wrs optimal}, which we prove as a corollary of our more general result for the partial inspection problem $\mathcal P(c)$, valid for all $c\in [0,2\pi]$ (see Theorem~\ref{thm: optimal worst case partial inspection sol}), rather than only for $c=2\pi/n$, $n\in \integers$. Our analysis also identifies optimal trajectories for arc lengths $c\in (2\pi/3,5\pi/6)$, a regime not covered in previous work. Moreover, coupling these worst-case solutions with our average-case upper bounds for $\mathcal P(c)$ yields nontrivial Pareto-optimal trade-offs. 

Finally, we note that the methodology we develop for the average case also recovers these worst-case results. By discretizing the set of perimeter points to be inspected and solving the resulting nonlinear program, we obtain trajectories that remain feasible in the continuous setting. Applying this to the trade-off objective
\[
\lambda \cdot \sup_P \inspect{P} + (1-\lambda) \cdot \mathbb{E}_P \inspect{P},
\]
with $\lambda=1$, reproduces the worst-case problem. The computer-assisted trajectory obtained in this way coincides with Isbell’s optimal solution, as shown in Figure~\ref{fig: trajectoryl1}.

\subsubsection{Improvements to the Average-Case Disk Inspection}

The first attempt to bound the average-case cost of \shoreline{1} is due to Gluss~\cite{gluss1961alternative}. His work proposed heuristic trajectories, which we revisit here in order to clarify a numerical error that led to under-reported bounds and to explain how his approach differs from ours. 

Gluss reported that the average-case performance of the Isbell-type algorithm (Figure~\ref{fig: isbel}) was $3.6276$, apparently for deployment angle $\theta = \pi/3$, without providing the supporting derivation. This value disagrees with the provable performance of Isbell-type trajectories established in Theorem~\ref{lem: isbell performance} (page~\pageref{lem: isbell performance}). Our calculations show that for $\theta=\pi/3$ the average-case cost equals $3.71386$, while optimization over all $\theta$ yields $3.70737$, which we believe are the correct values.

More significantly, the final numerical evaluations reported in~\cite{gluss1961alternative} for his improved heuristic trajectories are incorrect. Gluss described two heuristics, derived a common analytic expression for their expected cost (which we have verified), but in evaluating this expression numerically he introduced an error. He reported costs $3.4795$ and $3.4691$, respectively, both lower than our bound $3.5509015$ (Theorem~\ref{thm: general avg upper bound}). In fact, the correct evaluations of his formulas are $3.64529$ and $3.63489$, respectively, which are higher than our bound. Thus, our results do provide a genuine improvement, in addition to introducing a systematic methodology that produces a sequence of trajectories converging to the optimum. We next describe the heuristics and the source of the numerical error.

For the heuristics, Gluss was motivated by the worst-case optimal trajectory, which is tangent to the disk. He considered two variants (Figures~\ref{fig: gluss1},\ref{fig: gluss2}). The initial segment coincides with Isbell’s. From the origin, first follow segment $OA$ on line $x=1$ at angle $\theta$, then along $AB$ which is tangent to the disk at $B$. The path then follows the boundary counterclockwise from $B$ to $C=\mathcal P_{3\pi/2-\delta}$, with $\delta \in [0,\pi/2]$ (Isbell’s algorithm corresponds to $\delta=0$). Points $\mathcal P_\phi$ with $\phi \in [3\pi/2-\delta,2\pi]$ remain unvisited. In the first heuristic (Figure~\ref{fig: gluss1}), the trajectory continues along the tangent at $C$ up to point $D$ on the line $x=0$ (uniquely identified by $C$), then proceeds to $Z$ on $x=1$ via $D$. In the second heuristic (Figure~\ref{fig: gluss2}), the final segment $CDZ$ is replaced by a parabola smoothly tangent at $C$ and meeting $x=1$ orthogonally at $Z$. In both cases, the trajectory cost is optimized over the free parameters.

The analysis is summarized as follows. Since $C$ lies past $\mathcal P_\pi$, the average-case cost is
\[
\mathbb{E}_\phi \inspect{\mathcal P_{\phi}}
= \frac{1}{2\pi}\left(\int_{0}^{\pi} \inspect{\mathcal P_\phi}\, d\phi + \int_{\pi}^{2\pi} \inspect{\mathcal P_\phi}\, d\phi\right).
\]
Let $d$ denote the contribution of the final segment beyond $\mathcal P_\pi$. Then the cost can be expressed as
\[
\mathbb{E}_\phi \inspect{\mathcal P_{\phi}} = x(\theta) + d/2,
\]
where $d$ depends on the heuristic and $x(\theta)$ is given in~\cite{gluss1961alternative} as
\begin{equation}
\label{eq: gluss equation}
x(\theta)=
1
+\frac{1}{2\pi}
\log \left( \frac{1+\sinn{\theta}}{1-\sinn{\theta}} \right)
+(1 - \theta / \pi)( \sec(\theta) + \tan(\theta)- 1)
+
\frac{3\pi}4 - 2\theta +\frac{\theta(\theta-1)}{\pi}.
\end{equation}

The function $x(\theta)$ is minimized at $\theta_0=0.592334$ ($33^\circ~56'~17.52''$). Gluss reported the same minimizer (rounded to $33^\circ~56'$, or $0.592248$ radians), consistent with our verification. However, evaluating at $\theta_0$ yields $x(\theta_0)\approx 3.007998991$, so
\[
\mathbb{E}_\phi \inspect{\mathcal P_{\phi}} = 3.007998991 + d/2.
\]
Instead, Gluss reported $x(\theta_0)\approx 2.8422$. Notably, if the factor $(1-\theta/\pi)$ in~\eqref{eq: gluss equation} is (erroneously) replaced by $(1-2\theta/\pi)$, then at $\theta'=0.592248$ the evaluation gives $2.842402198$, matching his reported value to three decimals. This suggests a transcription error in his numerical evaluation.

Substituting the correct $x(\theta_0)$ into the formulas for his two heuristics yields corrected average costs $3.64529$ and $3.63489$ (rather than $3.4795$ and $3.4691$). 

In summary, our analysis corrects two distinct issues in~\cite{gluss1961alternative}. First, we establish the correct average-case performance of Isbell-type trajectories. The reported value $3.6276$ is inaccurate, and the true minimum is $3.70737$. Second, we identify a numerical error in the evaluation of Gluss’s heuristics. The reported values $3.4795$ and $3.4691$ should be $3.64529$ and $3.63489$, respectively. Beyond these corrections, our contributions differ from his in three essential ways. (i) Our bound of $3.5509015$ is strictly stronger than his corrected heuristic bound of $3.63489$. (ii) Our method is systematic. By discretizing the search space rather than the trajectory, we obtain a nonlinear program whose solution is guaranteed to be within $1+1/k$ of the true optimum for parameter $k$, under the assumption that the global optimizer induces trajectory not touching the disk. Although global optimality of the NLP solutions cannot be certified and computational limits restrict the choice of $k$, this approach yields robust bounds. (iii) Our evidence suggests that the optimal trajectory may not touch the disk boundary, in contrast to Gluss’s conjecture.

\section{Optimal Worst Case Inspection for \shoreline{n}}
\label{sec: worst case inspection for shoreline n agents}

\subsection{Some Preliminary Observations}
\label{sec: Some Preliminary Observations}

To simplify the notation, we remind the reader of some existing abbreviations and introduce new ones. Throughout our analysis, we will work with a fixed $k \in \mathbb{N}$, and $\theta, c \in [0, 2\pi]$, where $\theta < \pi/2$ and $\theta \leq 2c$. For such fixed $k$, $c$, and $\theta$, we define:
\begin{align}
    \phi_i & := 2\theta + \frac{c - 2\theta}{k} i, \quad i = 0, \ldots, k \label{def:phii} \\
    P_i & := \mathcal{P}_{\phi_i}, \quad i = 0, \ldots, k \label{def:pii}
\end{align}
The values $\phi_i$ are chosen such that the intervals $[\phi_i, \phi_{i+1}]$ form a (non-disjoint) partition of the interval $[2\theta, c]$, for $i = 0, \ldots, k - 1$. Additionally, the points $P_i$ represent the endpoints of the corresponding arcs on the perimeter of the unit disk.

For some $\phi \in [0, 2\pi]$, we also consider the parametric equation of the line $\mathcal{L}_{\phi}(t)$, which is tangent to the disk at the point $\mathcal{P}_{\phi}$:
\begin{equation}
\label{equa: parametric tangent line}
\mathcal L_{\phi}(t):=
\left(
\begin{array}{c}
\cos(\phi) \\
\sin(\phi) 
\end{array}
\right)
+
t
\left(
\begin{array}{c}
\sin(\phi) \\
-\cos(\phi)
\end{array}
\right).
\end{equation}
For the same fixed $k,c,\theta$, we will also need to consider the tangent $\mathcal L_{\phi}(t)$ at $\phi=\phi_i$, hence it is convenient to introduce the abbreviation 
$L_i(t) := \mathcal L_{\phi_i}(t), ~~t\in \reals, ~~i=0,\ldots,k.$

The following lemma summarizes some important observations pertaining to conditions that guarantee when points on the perimeter of the disk are inspected. 

\begin{lemma}\label{lem: preliminary inspections}
\begin{enumerate}[label=(\alph*), ref=\thelemma(\alph*)]
    \item \label{lem: coverage} An agent at point with polar coordinates $(r, \phi)$ and $r\geq 1$ inspects all points $\mathcal P_t$ that satisfy $|t- \phi| \leq \arccoss{1/r}$. 

    \item \label{lem: inspection range of first point}
    For every $\theta \in [0,\pi/2)$, an agent at point $(1,\tann{\theta})$ inspects all $\mathcal P_t$, with $t\in [0,2\theta]$. 

    \item \label{lem: inspection on tangent line} 
    An agent inspects point $\mathcal P_\phi$ if and only if it lies in the closed halfspace defined by $\mathcal L_\phi$ that does not contain the interior of the unit disk. 
In particular, an agent residing on any point of $\mathcal L_\phi$ inspects $\mathcal P_\phi$.
    
    
        \item \label{lem: inspection of sector}
        Consider points $A= \mathcal L_\phi(t_1)$ and $B=\mathcal L_{\phi+\gamma}(t_2)$ satisfying $\gamma<\pi/2$, $t_1\geq 0$ and $t_2\geq \frac{1-\coss{\chi}}{\sinn{\chi}}$. Then, for each $\phi \leq \chi \leq \chi+\gamma$, point $\mathcal P_\chi$ is inspected by some convex combination of $A,B$. In other words, the movement $A\rightarrow B$, inspects all points $\mathcal P_\chi$, with $\chi \in [\phi, \phi+\gamma]$. 
\end{enumerate}
\end{lemma}

\begin{proof}
\begin{figure}[h!]
    \centering
    \begin{subfigure}[t]{0.45\textwidth}
        \centering
        \includegraphics[width=4cm]{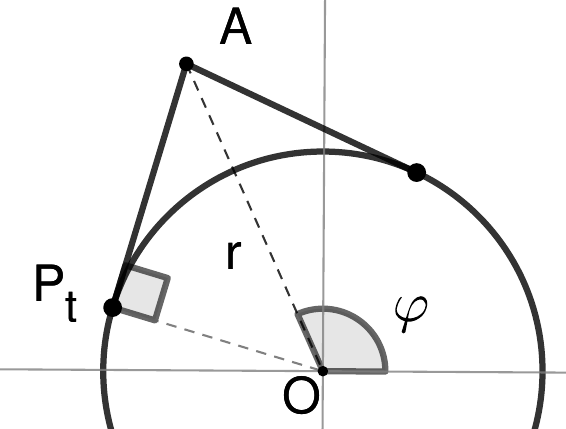}
        \caption{Supporting figure for Lemma~\ref{lem: coverage}.}
        \label{fig: pointpolar}
    \end{subfigure}
    \hfill
    \begin{subfigure}[t]{0.45\textwidth}
        \centering
        \includegraphics[width=4.8cm]{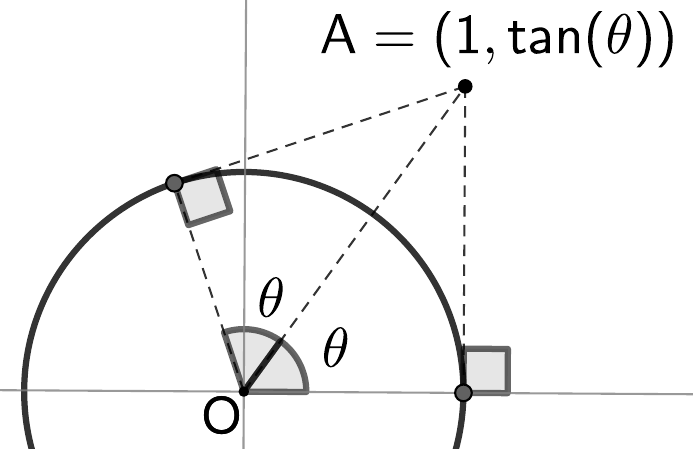}
        \caption{Supporting figure for Lemma~\ref{lem: inspection range of first point}.}
        \label{fig: initialpoint}
    \end{subfigure}
    
    \vspace{0.5cm}
    
    \begin{subfigure}[t]{0.45\textwidth}
        \centering
        \includegraphics[width=4cm]{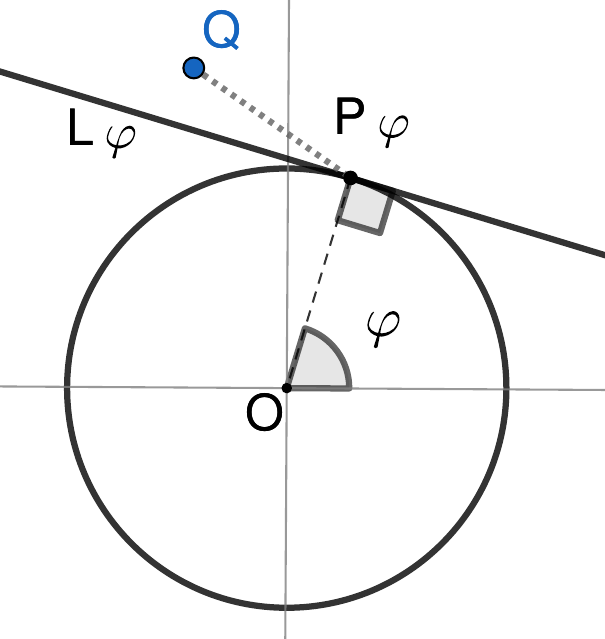}
        \caption{Supporting figure for Lemma~\ref{lem: inspection on tangent line}.}
        \label{fig: tangentline}
    \end{subfigure}
    \hfill
    \begin{subfigure}[t]{0.45\textwidth}
        \centering
        \includegraphics[width=5cm]{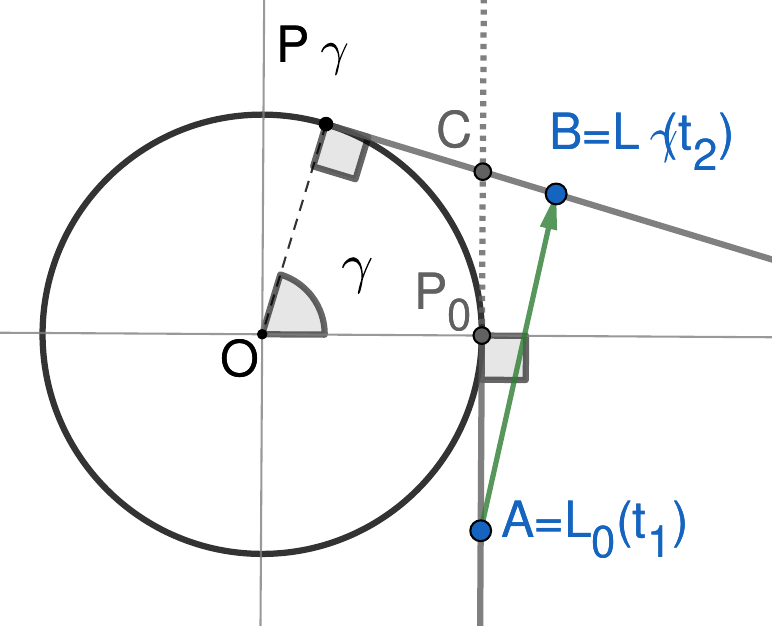}
        \caption{Supporting figure for Lemma~\ref{lem: inspection of sector}.}
        \label{fig: SectorInspection}
    \end{subfigure}
    \caption{Supporting figures for Lemma~\ref{lem: preliminary inspections}.}
    \label{fig: overall}
\end{figure}

\begin{enumerate}[label=(\alph*)]
\item See Figure~\ref{fig: pointpolar}.
A point with polar coordinates $(r, \phi)$ corresponds to Cartesian coordinates $A=r (\coss{\phi}, \sinn{\phi})$. if $r=1$, then $A$ is touching the perimeter of the disk, and it only inspects point $\mathcal P_\phi$. Otherwise, if $r>1$, we find the two lines that pass from $A$ and that are tangent to the disk. Then, all points in between the tangent points are inspected by point $A$. 

For this, consider an arbitrary point $\mathcal P_t$ on the disk. The line passing through $A,\mathcal P_t$ is tangent to the disk if and only if $\overrightarrow{O\mathcal P_t}=(\coss{t}, \sinn{t})$ is orthogonal to $\overrightarrow{\mathcal P_tA}=(r\coss{\phi} -\coss{t},r\sinn{\phi} -\sinn{t})$. But then, we see that 
$$
\overrightarrow{O\mathcal P_t}
\cdot 
\overrightarrow{\mathcal P_tA}
=
r\coss{\phi}\coss{t} + r\sinn{\phi}\sinn{t} -1 = r\coss{t-\phi}-1,
$$
and therefore the vectors are orthogonal exactly when $t=\phi \pm \arccoss{1/r}$, as wanted. 

\item See Figure~\ref{fig: initialpoint}.
 We have that $\|(1,\tann{\theta})\| = 1/\coss{\theta}$, which also implies that the angle of the polar coordinates of the point is $\theta$. Overall, by Lemma~\ref{lem: coverage}, an agent at polar coordinates $(1/\coss{\theta},\theta)$ inspects all points $\mathcal P_t$ satisfying $|t-\theta| \leq \arccoss{\coss{\theta}}$, that is all $t \in [0,2\theta]$. 

\item See Figure~\ref{fig: tangentline}.
 If an agent lies at point $Q$ in the closed halfspace defined by $\mathcal L_\phi$ not containing the interior of the unit disk, then no convex combination of $Q,\mathcal P_\phi$ is in the interior of the disk, hence she inspects $\mathcal P_\phi$.

For the other direction, assume that an agent inspects point $\mathcal P_\phi$ from point $Q$. Note that $\mathcal P_\phi$ cannot coincide with $Q$. 
If $Q$ is not in the closed halfsapce not containing the disk, then the ray (halfline) starting from $\mathcal P_\phi$ and passing though $Q$ must intersect the disk, contradicting that the agent inspects $\mathcal P_\phi$.

\item See Figure~\ref{fig: SectorInspection}. By rotating by $-\phi$, we may assume that $\phi=0$, hence $A= \mathcal L_0(t_1)$ and $B=\mathcal L_{\gamma}(t_2)$ for $t_1\geq 0$ and $t_2\geq \frac{1-\coss{\chi}}{\sinn{\chi}}$. 
By Lemma~\ref{lem: inspection on tangent line}, we have that $A$ inspects $\mathcal P_0$ and $B$ inspects $\mathcal P_\gamma$. Points $D$ that are convex combinations of $A,B$ inspect points $\mathcal P_\chi$, with $\chi \in [\phi, \phi+\gamma]$, unless $D$ is less than 1 away from the origin, or in other words, unless segment $AB$ intersects at two distinct points the perimeter of the circle (and hence containing a point in the interior of the circle). We identify the value of $t_2$ for which segment $AB$ intersects once the circle (at point $C$ in Figure~\ref{fig: SectorInspection}. 
For this we observe that the equation of line $\mathcal L_0(t)$ simplifies to $x=1$, whereas
$\mathcal L_{\gamma}(t) = (\coss{\chi}, \sinn{\chi})+ t( \sinn{\chi},-\coss{\chi})$. We conclude that $\mathcal L_{\gamma}(t)$ intersects $x=1$ when 
$(\mathcal L_{\gamma}(t))_1=1$, that is, when $\coss{\chi}+t\sinn{x}=1$, giving rise to the threshold promised in the statement of the lemma.  
\end{enumerate}
\qed  \end{proof}

By Lemma~\ref{lem: inspection on tangent line}, $\mathcal P_\phi$ is inspected if an agent lies in the halfspace defined by $\mathcal L_\phi(t)$ that does not contain the disk. 
Therefore, unless an agent is not already in the halfspace of $\mathcal L_\phi(t)$ not containing the disk, $\inspect{\mathcal P_\phi}$ is also equal to the first time an agent hits line $\mathcal L_\phi(t)$. 


\subsection{Solving \shoreline{n} Using the Partial Inspection Problem \pshoreline{c}. }
\label{sec: solve inspection by partial inspection}

Theorem~\ref{thm: optimal worst case partial inspection sol} (on page~\pageref{thm: optimal worst case partial inspection sol}) establishes the optimal worst-case cost for solving the partial inspection problem \pshoreline{c}. Clearly, \pshoreline{2\pi} is equivalent to the problem \shoreline{1}, and indeed, the case $c=2\pi$ and the corresponding optimal result (see Theorem~\ref{thm: original Isbell}, first proved in~\cite{isbell1957optimal}) is obtained as a corollary. First we provide the proof of Theorem~\ref{thm: optimal worst case partial inspection sol}, and then we present its corollaries pertaining to problems~\shoreline{n}.

\begin{proof}[of Theorem~\ref{thm: optimal worst case partial inspection sol}]
Without loss of generality, the points to be inspected are all points on the arc $\mathcal{A} = \{ \mathcal{P}_\phi : \phi \in [0, c]\}$. Hence, the optimal inspection trajectory is identified by the curve $C$ (of minimum length), with one endpoint at the origin $O$, and whose convex closure contains $\mathcal{A}$.

The reader may consult Figure~\ref{fig: convexcurveoverall}.
Since the agent initially lies at the origin, she does not inspect any points. Since $C$ is continuous and by Lemma~\ref{lem: inspection on tangent line}, the curve $C$ must intersect $\mathcal{L}_0$, say at point $A_0 = \mathcal{L}_0(t_0)$, and $\mathcal{L}_c$ at point $A_c = \mathcal{L}_c(t_c)$. Without loss of generality, we assume that $A_0$ is the first to be hit, and that $t_0 \leq 0$ and $t_c \geq 0$. Because $C$ is of minimum length, the curve must start with the line segment $0A_0$.

Next, we use the main technical contribution of~\cite{isbell1957optimal}, which states that the curve $C$ must be convex about the origin $O$. As a result, $C$ is a convex curve that starts with the segment $OA_0$, ends at point $A_c$. Moreover, the (sub)curve with endpoints $A_0$ and $A_c$, call it $C'$, is the minimum convex curve around $O$ that does not intersect the interior of the disk.

Let $D$ denote the point on $C'$ which is the last (moving from $A_0$ to $A_c$) to lie on the perimeter of the disk, or $D = A_0$ if no such point exists. Again, by the minimality (of length) of $C'$, $DA_c$ must be perpendicular to line $\mathcal{L}_c$. Finally, if $C'$ does intersect the perimeter of the disk, let $D'$ denote the first such point (moving from $A_0$ to $A_c$). Again, by the minimality of $C'$, the curve should start as the line segment $A_0D'$, which by the convexity of the curve $C'$ should not intersect the interior of the disk. Moreover, the segment $A_0D'$ should be tangent to the disk (at $D'$), as otherwise, one could choose a convex combination $A_0'$ of $A_0$ and $\mathcal{P}_\phi$ such that $A_0D'$ is tangent to the disk at $D'$. Furthermore, by the triangle inequality, the length of $OA_0' + A_0'D'$ would be less than that of $OA_0' + A_0D'$, contradicting the minimality of $C'$. To that end, note also that the portion of $C'$ between $D'$ and $D$ must lie on the perimeter of the disk forming an arc. 

Next we identify all curves (inspection trajectories) that comply with the above characteristics. We quantify point $A_0$ by $\theta \in [0,\pi/2)$, such that $A_0 = (1,\tann{\theta})$, where also $\theta \leq c/2$. We distinguish two cases as to whether curve $C'$ has at most one point on the perimeter of the disk (hence $C$ has at 1 or 2), or infinitely many. For each of these cases, we describe a family of curves, type-1 and type-2, parameterized by $\theta$, and applicable only under conditions of $c,\theta$.

\emph{Type-1 curves} (see~Figure~\ref{fig: convexcurve1}) intersect the perimeter of the disk in 1 or 2 points. Hence for point $A_0 = (1,\tann{\theta})$ and the end point $A_c \in \mathcal L_c$ of $C$, we have that $A_0A_c$ is perpendicular to $\mathcal L_c$. This implies that $0\leq c \leq 3\pi/2$, and that $2\theta \geq c-\pi/2$. The cost of that trajectory equals $OA_0+A_0A_c$, where $OA_0=1/\coss{\theta}$. In order to find $A_0A_c$, draw a line parallel to $\mathcal L_c$, passing through the origin, intersecting $A_0A_C$ at point $K$. Clearly, $KA_c=1$, in triangle $OKA_0$, we have that $\sinn{c-\theta-\pi/2}=A_0K/OA_0$, and hence $A_0K=\sinn{c-\theta}/\coss{\theta}$. Overall, the length of that trajectory becomes 
$$
f_1(\theta) := 1/\coss{\theta}+\sinn{c-\theta-\pi/2}/\coss{\theta}+1.
$$

\emph{Type-2 (Isbell-type) curves} (see~Figure~\ref{fig: convexcurve2}) intersect the perimeter of the disk in infinitely many points. In other words, trajectory is made up by line segment $OA_0$, then tangent line segment $A_0D'$ (where $D'=\mathcal P_{2\theta}$), then arc $D'D$, and the line segment $DA_c$ (tangent to the disk at point $D=\mathcal P_{c-\pi/2}$ and) perpendicular to $\mathcal L_c$. For this configuration to be well defined, we need $\theta \leq c/2+\pi/2$, as well as $c\geq \pi/2$. Note that $A_0D'=A_0P_0 = \tann{\theta}$, the length of the arc $D'D$ equals $c-\pi/2-2\theta$, while $DA_c=1$. But then, the length of the curve, becomes 
$$
f_2(\theta) = 1/\coss{\theta}+\tann{\theta}+c-\pi/2-2\theta+1.
$$

It follows that when $c\leq \pi/2$ only type-1 curves are applicable, when $\pi/2\leq c \leq 3\pi/2$ both type-1 and type-2 curves are applicable, and when $c\geq 3\pi/2$ only type-2 curves are applicable. 

Lemmata~\ref{lem: minimizer of f1} and~~\ref{lem: minimizer of f2} in Section~\ref{sec: proofs from sec solve inspection by partial inspection} minimize $f_1(\theta)$ and $f_2(\theta)$. Therefore, for each $c\in [0,\pi/2]$ for which only type-1 curves are applicable, their optimal length is $f_1(c/2) = 1/\coss{c/2}$. Similarly, for each $c\in [3\pi/2,2\pi]$ for which only type-2 curves are applicable, their optimal length is $f_2(\pi/6) = \sqrt3+c-5\pi/6+1$. For the remaining values of $c \in [\pi/2,3\pi/2]$, we need to compare $f_1,f_2$, which is done in Lemma~\ref{lem: f1-f2 comparison} in Section~\ref{sec: proofs from sec solve inspection by partial inspection}, concluding the statement of the theorem. 
\qed  \end{proof}

\begin{figure}[h!]
    \centering
    \begin{subfigure}[t]{0.45\textwidth}
        \centering
    \includegraphics[width=6cm]{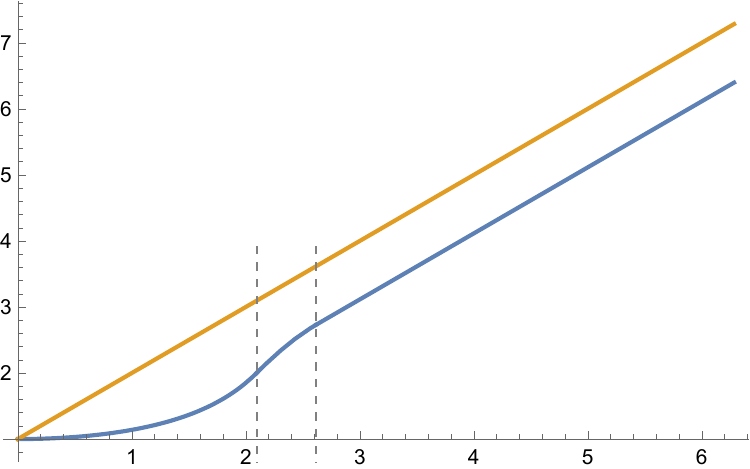}
    \caption{
The worst-case optimal inspection cost for solving the partial problem \pshoreline{c} (blue curve), as a function of $c$ (see Theorem~\ref{thm: optimal worst case partial inspection sol}). The dotted lines represent the values of $c$ at which the piecewise-defined cost function changes its formula. The yellow curve is given for comparison, and represents $1+c$ which is the cost of the naive algorithm for solving \pshoreline{c}. 
}
\label{fig: partialperformance}
    \end{subfigure}\hfill
    \begin{subfigure}[t]{0.45\textwidth}
        \centering
    \includegraphics[width=6cm]{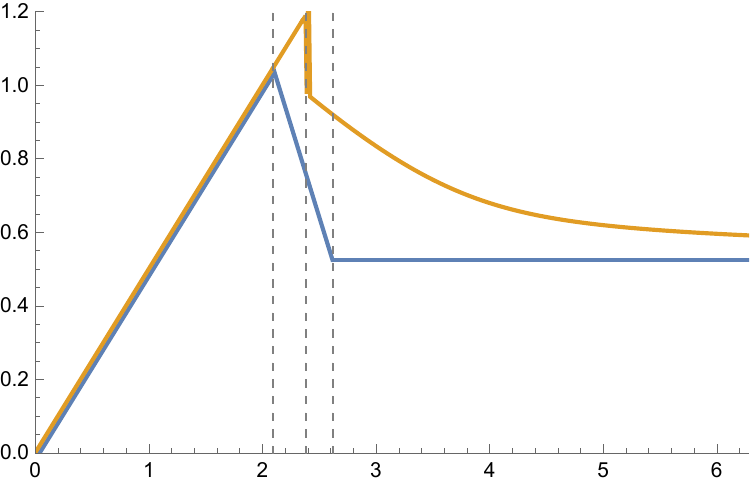}
    \caption{
The worst-case optimal deployment angle $\theta$ (in blue) of \textsc{EPS}$(\theta,t_1,\ldots,t_k)$ is shown as a function of the arc length $c$ to be inspected in \pshoreline{c} (see Theorem~\ref{thm: optimal worst case partial inspection sol}). 
The yellow plot illustrates the behavior of the angle $\theta=\theta(c)$ for the inspection algorithm \textsc{EPS}$(\theta,t_1,\ldots,t_k)$, when the parameters are chosen to optimize the average inspection cost for \pshoreline{c}, see Theorem~\ref{thm: upper bounds to average inspection  partial}.
The gray dotted lines represent transition thresholds for the optimizers $\theta$ in the two cases. 
}
\label{fig: worstcaseOptDeploymentAngle}
    \end{subfigure}\hfill
    \caption{Visualization of the results in Theorem~\ref{thm: optimal worst case partial inspection sol}, and Theorem~\ref{thm: upper bounds to average inspection  partial}.
}
    \label{fig: upper bound for partial inspection}
\end{figure}

A visualization of Theorem~\ref{thm: optimal worst case partial inspection sol} is provided in Figure~\ref{fig: upper bound for partial inspection}. 
To prove Theorem~\ref{thm: optimal worst case partial inspection sol} for \pshoreline{c}, we identify the shortest curve, $C$, that inspects all points along the arc of length $c$. The proof distinguishes between two types of curves: Type-1 (Figure~\ref{fig: convexcurve1}), which intersects the disk perimeter at one or two points and applies when $0 \leq c \leq 3\pi/2$, and Type-2 (Isbell-type) curves (Figure~\ref{fig: convexcurve2}), which intersect the perimeter infinitely and apply when $c \geq \pi/2$. For each configuration, we derive conditions for minimal length based on convexity constraints around the disk, ensuring that $C$ fully covers the arc while minimizing total distance. The detailed derivation of these cases is provided in the appendix.

\begin{figure}[h!]
    \centering
    \begin{subfigure}[t]{0.45\textwidth}
        \centering
        \includegraphics[width=5cm]{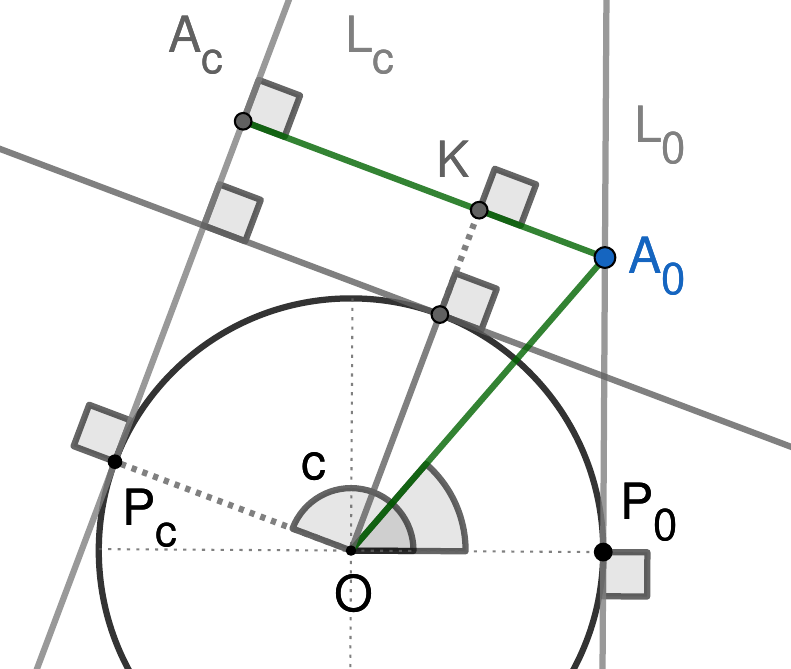}
        \caption{Type-1 curves (in green) for solving \pshoreline{c}, applicable when $0 \leq c \leq 3\pi/2$ and $2\theta \geq c-\pi/2$.}
        \label{fig: convexcurve1}
    \end{subfigure}
    \hfill
    \begin{subfigure}[t]{0.45\textwidth}
        \centering
        \includegraphics[width=4.8cm]{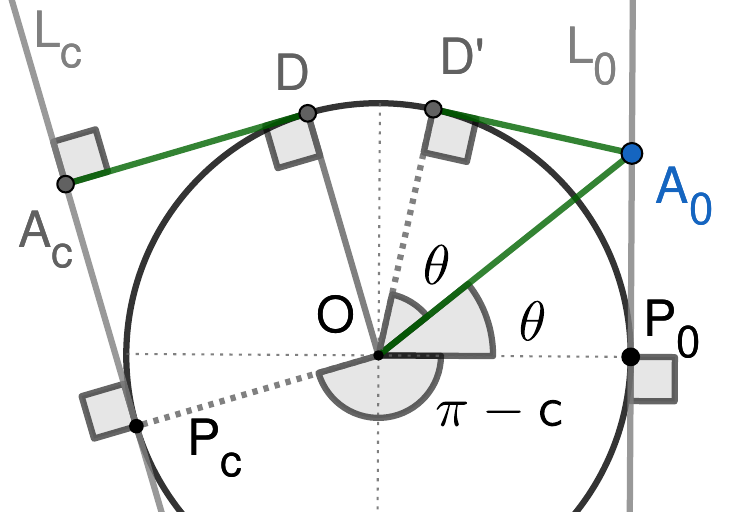}
        \caption{Type-2 (Isbell-type) curves (in green) for solving \pshoreline{c}, applicable when $\pi/2\leq c \leq 2\pi$ and $2\theta \leq c-\pi/2$.}
        \label{fig: convexcurve2}
    \end{subfigure}
    \caption{Trajectories achieving the worst-case optimal inspection cost to \pshoreline{c} as shown in Theorem~\ref{thm: optimal worst case partial inspection sol}.}
    \label{fig: convexcurveoverall}
\end{figure}

We note that the optimal worst case inspection cost to \pshoreline{c} is continuous at $c=2\pi/3$ and at $c=5\pi/6$ with corresponding worst case costs equal to $2$ and $1+\sqrt3$, respectively. Moreover, and not surprisingly, it is easy to see that the cost is monotonically increasing in $c$. 

As an immediate corollary we obtain the optimal worst case cost for the general inspection problem \shoreline{n}, as stated in Theorem~\ref{thm: general wrs optimal}.

\begin{proof}[of Theorem~\ref{thm: general wrs optimal}]
Consider any partition of the perimeter of the disk that assigns points to inspect to the $n$ agents. This defines a number of arcs covering the perimeter of the disk, each labeled according to the assigned agents. 
Arcs in the partition can be rearranged (each agent still being assigned the same total portion of the disk) so that those with the same label are contiguous, without increasing the cost.

Hence, each agent needs to inspect a contiguous arc. Next, we claim that all agents inspect arcs of the same length. If not, we could slightly decrease the largest arc and slightly increase the smallest arc. From Theorem~\ref{thm: optimal worst case partial inspection sol}, the optimal cost for covering any arc of length $c$ is monotonically increasing in $c$, and thus this adjustment would reduce the overall inspection cost.

We conclude that each agent inspects an arc of length $c = 2\pi/n$. The claim of the theorem follows from Theorem~\ref{thm: optimal worst case partial inspection sol}, which describes the worst-case optimal cost for inspecting an arc of length $c$, for every $c \in [0, 2\pi]$.
\qed  \end{proof}

\section{Average-Case Inspection Upper Bounds for \shoreline{n}}
\label{sec: average case inspection for shoreline n agents}

\subsection{Discrete Inspection Algorithms for \ashoreline{k}{\theta}{c} and \pshoreline{c}}
\label{sec: discrete average}

In this section, we introduce and analyze the trajectory (Algorithm) \textsc{PolySegment} for solving \ashoreline{k}{\theta}{c} and the trajectory (Algorithm) \textsc{ExtendedPolySegment}, which, under certain conditions, will solve the 1-agent partial inspection problem \pshoreline{c} and eventually the $n$-agent inspection problem \shoreline{n}. 

Recall that $\theta \in [0, \pi/2)$, and $c \in [0, 2\pi]$, with $c \geq 2\theta$. We fix these constants and also fix $k \in \naturals$, i.e., for the description of the trajectory, we require $k$ to be finite. The algorithm will be parameterized by $(t_1, \ldots, t_k)$, where $t_i \in \reals$. Recall that in \ashoreline{k}{\theta}{c}, we need to inspect points $P_i$ with angular values
$
\phi_i := 2\theta + \frac{c - 2\theta}{k} i, \quad i = 0, \ldots, k.
$
Note that the remaining of the necessary notation was introduced earlier in Section~\ref{sec: Some Preliminary Observations}.
Along with the algorithm description below, the reader may also consult Figure~\ref{fig: PolySegmentAlgorithm} on page~\pageref{fig: PolySegmentAlgorithm}.

\begin{figure}[h!]
    \centering
    \includegraphics[width=0.6\textwidth]{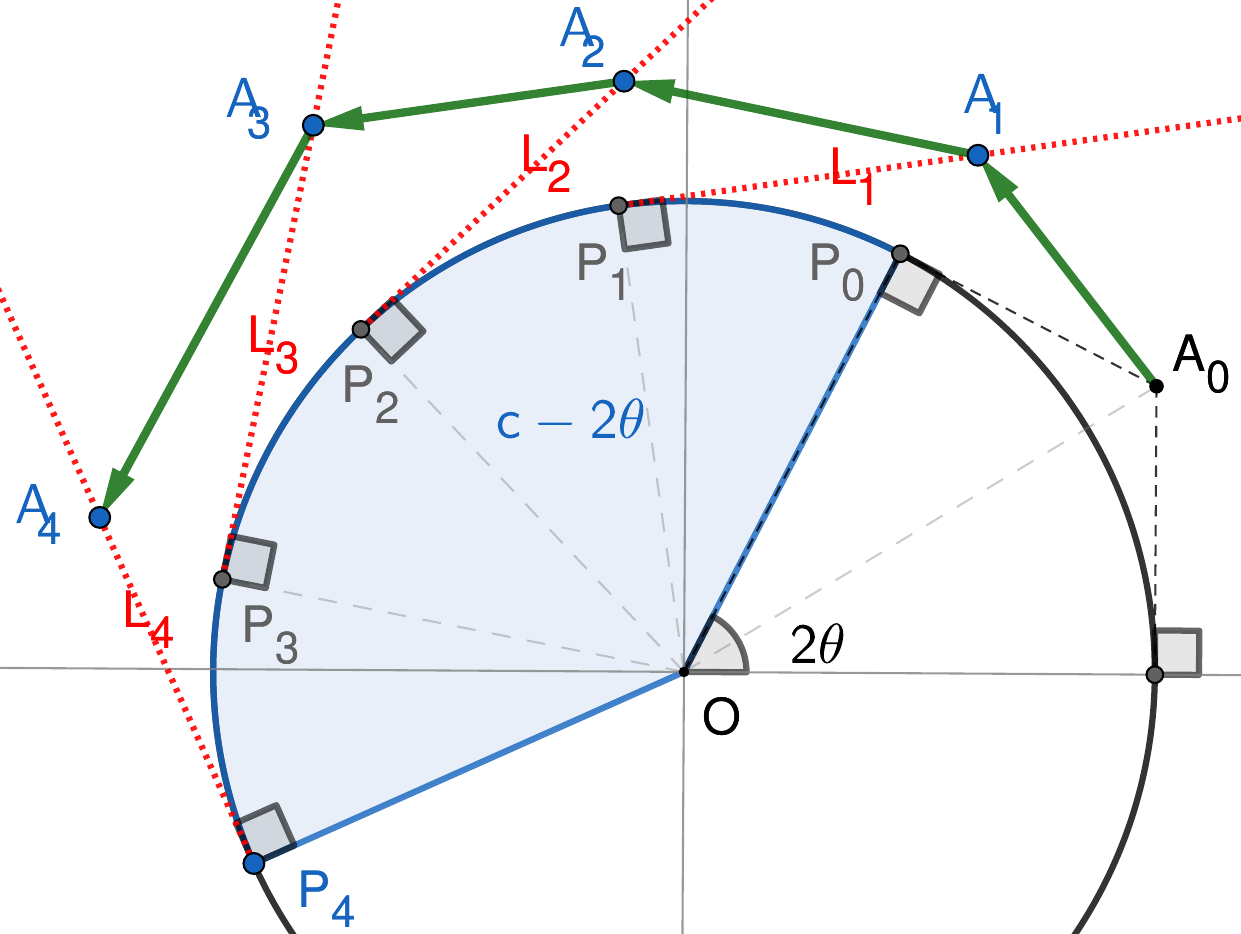}
    \caption{
        In this example $k=4$, and we demonstrate how \textsc{PS}$(t_1,\ldots,t_4)$ is feasible to \ashoreline{4}{\theta}{c}. 
        The disk segment of arc length $c-2\theta$, whose perimeter needs to be inspected, is depicted as the light blue sector. 
        The tangent lines $L_i(t)$ are depicted as dotted red lines, where only the halflines corresponding to $t\geq0$ are shown. 
        We let $A_0=(1,\tann{\theta})$, and we note that there exists $t_0$ such that $A_0 = \mathcal L_{\phi_0}(t_0)$, where $\phi_0=2\theta$. 
        For $i=1,2,3,4$, the points $A_i$ are obtained as $A_i=L_i(t_i)$ with $t_i>0$. 
        By inspection, we see that the choice of $t_i$ makes the trajectory $A_0\rightarrow A_1\rightarrow A_2\rightarrow A_3 \rightarrow A_4$ feasible to 
        \ashoreline{4}{\theta}{c} as well as to \ashoreline{\infty}{\theta}{c}. Feasibility to \ashoreline{\infty}{\theta}{c} will be examined later, and we stress that it is not sufficient to have $t_i>0$. 
    }
    \label{fig: PolySegmentAlgorithm}
\end{figure}

\begin{aalgorithm}[\textsc{PolySegment}$(t_1,\ldots,t_k)$]
Given $k,\theta,c$ with $\theta <\pi/2$, $c\geq 2\theta$, we set $A_0:=(1,\tann{\theta})$, and $A_i := L_{i}(t_i)$. Then, the trajectory is defined by the movement between points 
$A_0\rightarrow A_1\rightarrow \ldots \rightarrow A_k$. 
\end{aalgorithm}

For fixed $k,\theta,c$, our goal will be to determine optimal values for $(t_1,\ldots,t_k)$ so as to minimize the average case performance to \ashoreline{k}{\theta}{c}. For the ease of notation, we will refer to the trajectory as \textsc{PS}$(t_1,\ldots,t_k)$.
Next we comment on the algorithms feasibility and compute its average case inspection cost.

\begin{lemma}
\label{lem: polysegment performance for discrete}
For every $k \in \naturals$ and for all $t_i\in \reals$,  trajectory \textsc{PS}$(t_1,\ldots,t_k)$ is feasible to \ashoreline{k}{\theta}{c}.
Moreover, 
trajectory \textsc{PS}$(t_1,\ldots,t_k)$ has worst case cost 
$\sum_{j=1}^k \|A_{j-1}A_j\|$ 
and average case cost 
$$\frac{1}{k+1} \sum_{j=1}^k (k-j+1) \|A_{j-1}A_j\|.$$
\end{lemma}

\begin{proof}
By Lemma~\ref{lem: inspection on tangent line}, an agent on $L_i(t)$ inspects point $P_i$, and we have one such $A_i=L_i(t_i)$, for each $i=0,\ldots,k$. Therefore, trajectory \textsc{PS}$(t_1,\ldots,t_k)$ is indeed feasible. 

The length of the trajectory equals $\sum_{j=1}^k \|A_{j-1}A_j\|$, and since every point $P_i$ is inspected no later than when the agent reaches $A_i$, it follows that no point is inspected later than the length of the trajectory, concluding the claim about the worst case. 

Next we calculate the average inspection cost of points $P_0, P_1,\ldots,P_k$, deriving this way the average case analysis. For this, we calculate the inspection time $\inspect{P_i}$ for each $P_i$, and recall that by Lemma~\ref{lem: inspection range of first point} we have that $\inspect{P_0}=0$. 
Also, we see that $\inspect{P_i} \leq \sum_{j=1}^i \|A_{j-1}A_j\|$, for $i=1,\ldots,k$ and hence the average inspection time is 
$$
\frac{1}{k+1}
\sum_{i=0}^k \inspect{P_i}
\leq
\frac{1}{k+1}
\sum_{i=1}^k \sum_{j=1}^i \|A_{j-1}A_j\|
= 
\frac{1}{k+1}
\sum_{i=1}^k (k-i+1) \|A_{i-1}A_i\|.
$$
\qed  \end{proof}

Trajectory \textsc{PolySegment}$(t_1,\ldots,t_k)$ will be eventually used to derive upper bounds to \shoreline{n}. Below we extend the trajectory, and we quantify its performance assuming it is feasible for the continuous partial inspection problem. In that direction, we present trajectory (algorithm) \textsc{ExtendedPolySegment} for solving \pshoreline{c}.

\begin{aalgorithm}[\textsc{ExtendedPolySegment}$(\theta,t_1,\ldots,t_k)$]
\label{alg: extendedpolysegment} ~ \\
Given $c\in [0,2\pi], k\in \naturals$ and $\theta\leq \{c/2,\pi/2-\epsilon\}$ for some $\epsilon>0$, we set 
$A_0:=(1,\tann{\theta})$, and $A_i := L_{i}(t_i)$.\footnote{Recall that the definition of lines $L_i$ is tailored to the fixed $k,\theta,c$.} 
Then, the trajectory of the agent is defined by the movement between points 
$O\rightarrow A_0\rightarrow A_1\rightarrow \ldots \rightarrow A_k$, where $O$ is the center of the disk. 
\end{aalgorithm}

For notational convenience, we will refer to the algorithm above as \textsc{EPS}$(\theta,t_1, \ldots, t_k)$. Next, we describe its performance on \pshoreline{c} as a function of its performance on the discrete problem, assuming also that it is feasible for the problem. We will introduce the feasibility conditions later.

\begin{lemma}
\label{lem: reduction from discrete to continuous}
Let $k\in \naturals$ and $\theta\leq \{c/2,\pi/2-\epsilon\}$ for some $\epsilon>0$.
Suppose also that \textsc{PS}$(t_1,\ldots,t_k)$\footnote{Note that the trajectory points are tailored to the fixed values of $k,\theta,c$.} is feasible to \ashoreline{\infty}{\theta}{c}
with worst case cost and the average case inspection cost equal to $s_1$ and $s_2$, respectively.  Then, 
\begin{enumerate}[label=(\alph*), ref=\thelemma(\alph*)]
\item 
\label{lem: reduction from discrete to continuous feasibile}
\textsc{EPS}$(\theta,t_1,\ldots,t_k)$ is feasible to \pshoreline{c}.
\item 
\label{lem: reduction from discrete to continuous wrs}
The worst case cost of \textsc{EPS}$(\theta,t_1,\ldots,t_k)$ to \pshoreline{c} is $\frac{1}{\coss{\theta}} + s_1$. 
\item 
\label{lem: reduction from discrete to continuous avg}
The average case cost of \textsc{EPS}$(\theta,t_1,\ldots,t_k)$ to \pshoreline{c} is 
$$
\frac{1}{c}
\log \left( \frac{1+\sinn{\theta}}{1-\sinn{\theta}} \right)
+
\left(1 - \frac{2\theta}{c} \right) \left(\frac{1}{\coss{\theta}} + s_2 \right).
$$
\end{enumerate}
\end{lemma}

\begin{proof}
\begin{enumerate}[label=(\alph*)]
\item We are given that \textsc{PS}$(t_1,\ldots,t_k)$ is feasible to \ashoreline{\infty}{\theta}{c}, therefore, by definition all points $\mathcal P_\phi$, with $\phi \in [2\theta,c]$, are inspected. Trajectory \textsc{PS}$(t_1,\ldots,t_k)$ has the additional movement $0\rightarrow A_0$, where $A_0=(1,\tann{\theta})$. By Lemma~\ref{lem: inspection range of first point}, once the agent reaches $A_0$, all points $\mathcal P_\phi$, with $\phi \in [0,2\theta]$ are inspected as well. Effectively, the first agent inspects all points along an arc of total length $c$.

\item 
The worst case performance of \textsc{EPS}$(\theta,t_1,\ldots,t_k)$ is just the length of the trajectory, which equals $\|0A_0\|$ plus the length of the trajectory of \textsc{EPS}$(\theta,t_1,\ldots,t_k)$. The claim follows by observing that $\|0A_0\|=\sqrt{1^2+\tan^2(\theta)} = 1/\coss{\theta}$, where the last equality is due to that $\theta\in [0,\pi/2)$. 

\item 
We are give that the average inspection cost to the auxiliary inspection problem \ashoreline{\infty}{\theta}{c} equals $s_2$, and that pertains to the cost of inspecting points 
$\mathcal P_{\phi}$ (in the partial problem) where $\phi \in [2\theta,c]$. For the inspection problem $\mathcal P(c)$, the inspection cost of the same points incur the additional cost of $1/\coss{\theta}$, which is the length of the initial deployment segment $OA_0$, see Figure~\ref{fig: initialpointavg}. 
Overall, this shows that 
$$
E_{\phi \in [2\theta,c]}[\inspect{\mathcal P_\phi} ] 
= \frac{1}{\coss{\theta}} + s_2
$$
\begin{figure}[h!]
    \centering
    \begin{minipage}{5.0cm}
        \includegraphics[width=\linewidth]{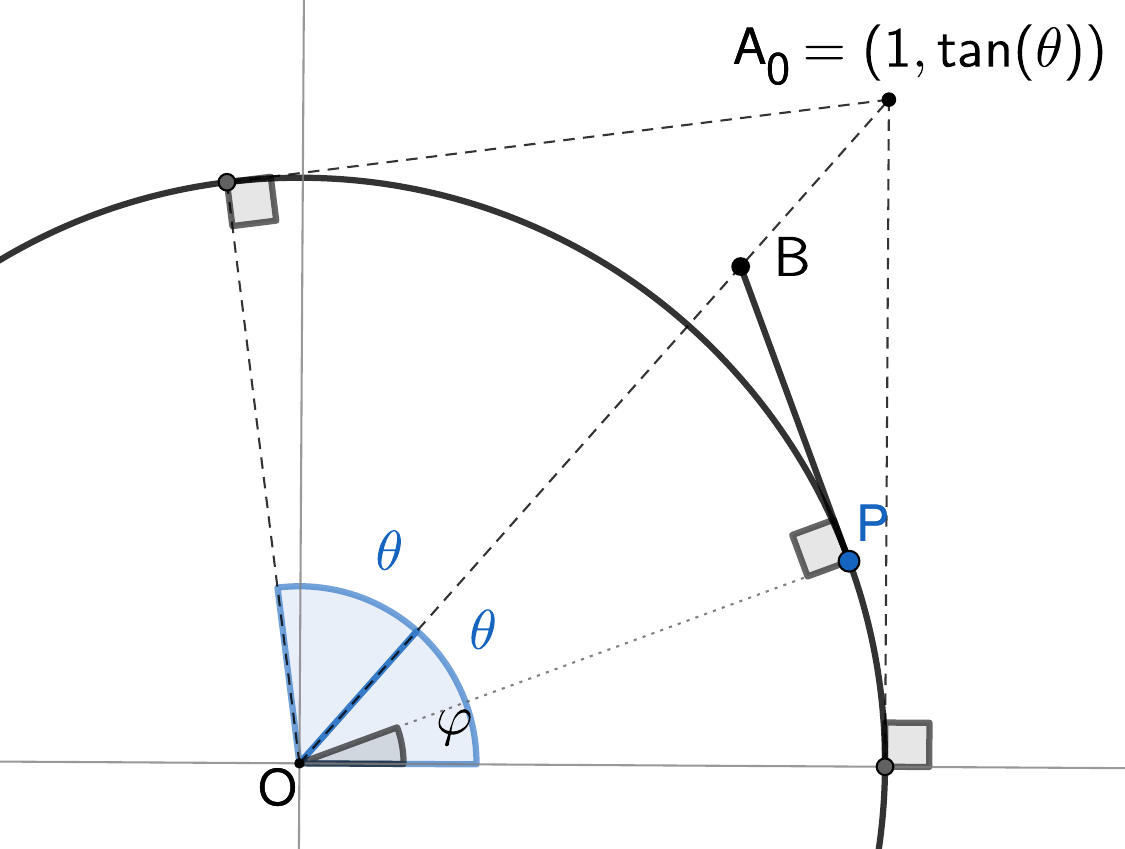}
    \end{minipage}%
    \hspace{1cm}%
    \begin{minipage}{0.55\textwidth}
        \caption{
        The inspection time of points $\mathcal P_\phi$, where $\phi \in [0,\theta]$, in the initial movement $O\rightarrow A_0$ of the trajectory \textsc{ExtendedPolySegment}.
        }
        \label{fig: initialpointavg}
    \end{minipage}
\end{figure}

Next, and for each $\phi \in [0,\theta]$ we calculate the inspection time $\inspect{\mathcal P_\phi}$ of $\mathcal P_\phi$, depicted as point $P$ in Figure~\ref{fig: initialpointavg}. That is, $B$ is the intersection of the line tangent on the disk at $\mathcal P_\phi$ with $OA_0$. 
For the analysis, let's focus first on the case that $\phi \leq \theta$. Note that by Lemma~\ref{lem: inspection on tangent line}, we have that 
$\inspect{\mathcal P_\phi} = \|OB\| = 1/\coss{\theta-\phi}$, where  the last equality is due to the right triangle $OPB$. By symmetry, the average inspection of points $\mathcal P_\phi$ is the same when $\phi \in [0,\theta]$ and when $\phi \in [\theta,2\theta]$. Therefore, the overall average inspection time is
\begin{align*}
E_{\phi \in [0,2\pi/n]}[\inspect{\mathcal P_\phi} ]
& =  \frac{1}{c} \int_{0}^{2\pi/n} \inspect{\mathcal P_\phi} \d{\phi}  \\
& = 
\frac{1}{c}
\left( 
\int_{0}^{2\theta} \inspect{\mathcal P_\phi} \d{\phi} 
+
\int_{2\theta}^{c} \inspect{\mathcal P_\phi} \d{\phi} 
\right) \\
& = 
\frac{1}{c}
\left( 
2 \int_{0}^{\theta} \frac{1}{\coss{\theta-\phi}}  \d{\phi} 
+ 
\left(c - 2\theta \right) s_2
\right)
\end{align*}
The claim follows by observing that 
$$\int_{0}^{\theta} \frac{1}{\coss{\theta-\phi}} \d{\phi} = \frac12 \log \left( \frac{1+\sinn{\theta}}{1-\sinn{\theta}}  \right),$$
which is based on that $\int \frac1{\coss{\phi}} \d \phi = \textrm{ArcCoth}( \sinn{\phi})$. 
\end{enumerate}
\qed  \end{proof}


We are now ready to introduce the feasibility conditions of the discrete inspection algorithm for the continuous inspection problems, as well as to relate the performance of the algorithm for the discrete and continuous inspection problems.
As a reminder, the parameter $k$ quantifies the refinement of the continuous space into a discrete set of points to be inspected. The value of $k$ that we will eventually use is on the order of thousands, and the higher it is, the better our results for the continuous problem. The requirement in the lemma that $k \geq 5$ becomes apparent only in the proof, where we need the points to be inspected to be strictly less than $\pi/2$ apart.

\begin{lemma}
\label{lem: polysegment performance for continuous avg}
Consider $k \geq 5 \in \naturals$ and $t_i\in \reals$, satisfying 
$t_i \geq \tfrac{1-\coss{\frac{c-2\theta}k}}{ \sinn{\frac{c-2\theta}k} }, i=1,\ldots,k$. 
Then, trajectory \textsc{PS}$(t_1,\ldots,t_k)$ is feasible to \ashoreline{\infty}{\theta}{c}.
Moreover, if $s$ is the average case cost for the same trajectory to \ashoreline{k}{\theta}{c}, then the average case cost to \ashoreline{\infty}{\theta}{c} is at most $\left(1+\frac{1}{k}\right) s$. 
\end{lemma}

\begin{proof}
Fix $k \geq 5$ and $t_1,\ldots,k$ 
satisfying $t_i \geq \tann{\frac{c-2\theta}k}, i=1,\ldots,k$
along with the trajectory \textsc{PS}$(t_1,\ldots,t_k)$ which can be represented as $A_0\rightarrow A_1\rightarrow \ldots \rightarrow A_k$. 
First, we argue that the trajectory is feasible to \ashoreline{\infty}{\theta}{c}. For this, note that if $c=2\theta$, then by Lemma~\ref{lem: inspection range of first point} the only point that needs to be inspected, i.e. $\mathcal P_{2\theta}$, is indeed inspected due to $A_0$. Hence, we may assume that $c>2\theta$. 
We need to argue that all points $\mathcal P_\phi$, where $\phi \in [2\theta,c]$, are inspected. For this, observe that the interval can be partitioned into the (non-disjoint) intervals $[\phi_{i-1}, \phi_{i}]$ with $\phi_{i} - \phi_{i-1} = \frac{c-2\theta}k \leq \frac{2\pi}5<\pi/2$. 
Therefore by Lemma~\ref{lem: inspection of sector}, each such subinterval is inspected by the movement $A_{i-1}\rightarrow A_{i}$, where $i=1,\ldots,k$.
Hence the trajectory is feasible to \ashoreline{\infty}{\theta}{c}.

Next we provide an upper bound to the cost of the trajectory, as a function of the average case cost of the same trajectory to \ashoreline{k}{\theta}{c}. For this, we recall from the argument above that for every $\phi \in [\phi_{i-1},\phi_i]$, point $\mathcal P_{\phi}$ is inspected during the movement $A_{i-1}\rightarrow A_{i}$, ad hence for these points we have that $\inspect{\mathcal P_\phi}\leq \sum_{j=1}^i \|A_{j-1}A_{j}\|$. Noting also that $\inspect{\mathcal P_\phi}>0$ for all $\phi \in [2\theta,c]$, allows to calculate the average cost of the trajectory to \ashoreline{\infty}{\theta}{c} (as a function of the average case performance $s$ to \ashoreline{k}{\theta}{c}) as
\begin{align}
\notag \frac{1}{c-2\theta} \int_{2\theta}^c \inspect{\mathcal P_\phi} \d{\phi}
&= 
\frac{1}{c-2\theta} 
\sum_{i=1}^{k} \int_{\phi_{i-1}}^{\phi_i} \inspect{\mathcal P_\phi} \d{\phi} \\
\notag
&\leq  
\frac{1}{c-2\theta} 
\sum_{i=1}^{k} \int_{\phi_{i-1}}^{\phi_i} \sum_{j=1}^i \|A_{j-1}A_{j}\|  \d{\phi} \\
\notag
&= 
\frac{1}{c-2\theta} 
\sum_{i=1}^{k} \left( \phi_{i} - \phi_{i-1} \right) \sum_{j=1}^i \|A_{j-1}A_{j}\|   \\
&= 
\frac{1}{c-2\theta} 
\sum_{i=1}^{k} \left( \phi_{k} - \phi_{i-1} \right) \|A_{i-1}A_{i}\|.
\label{equa: upper bound to avg using discrete}
\end{align}
Finally, we see that 
$$
\phi_{k} - \phi_{i-1}
=
c - \left(
2\theta+\frac{c-2\theta}{k} (i-1)
\right)
= 
(c-2\theta)
\left(
1-\frac{i-1}{k}
\right)
$$
Therefore, continuing from~\eqref{equa: upper bound to avg using discrete}, we see that an upper bound to average case performance to \ashoreline{\infty}{\theta}{c} is 
$$
\sum_{i=1}^{k}  \left(
1-\frac{i-1}{k}
\right) \|A_{i-1}A_{i}\|
=
\frac{1}{k}\sum_{i=1}^{k}  \left(k-i+1
\right) \|A_{i-1}A_{i}\|
= \frac{k+1}{k} s,
$$
where the last equality follows from Lemma~\ref{lem: polysegment performance for discrete}.
\qed  \end{proof}

\subsection{Upper Bounds to the Average Case Inspection of \pshoreline{c} and \shoreline{n}}
\label{sec: solution to avg partial and shoreline-n}

In this section we introduce a Non Linear Program~\eqref{NLP-avg} for finding a feasible \textsc{ExtendedPolySegment} trajectory to \pshoreline{c} and for minimizing its average case performance. The algorithm also takes as input some $k \geq 5$ and $c \in [0, 2\pi]$, ensuring that the points to be inspected on the disk's perimeter are less than $\pi/2$ apart. In fact, the values of $k$ we use are on the order of thousands, and the higher they are, the better our results for the continuous problem.

For fixed $k,\theta$, and using the notation introduced at the beginning of Section~\ref{sec: Some Preliminary Observations}, we write explicit formulas for points $A_0,A_1, \ldots,A_k$ used by algorithm \textsc{EPS}$(\theta,t_1, \ldots, t_k)$ (see the formal description in~\ref{alg: extendedpolysegment}) with respect to the NLP variables $\theta \in \reals$ and $t\in \reals^k$. For this, we also recall that $A_i = L_{\phi_i}(t_i)$ for $i=1,\ldots,k$, and that $A_0=(1,\tann{\theta})$. 

Therefore, it is easy to see that 
\begin{align*}
\|A_0A_1\|^2 &=
(-1 + c_{\theta,k}' +  s_{\theta,k}' t_1)^2 + (-s_{\theta,k}' + c_{\theta,k}' t_1 + \tan(\theta))^2, \\
\|A_iA_{i+1}\|^2 &=2 - 2 c_{\theta,k} + 2 s_{\theta,k} t_i + t_i^2 - 2 (s_{\theta,k} + c_{\theta,k} t_i)  t_{i+1} + t_{i+1}^2,
\end{align*}
where
\begin{align*}
c_{\theta,k}' &= \cos((2  (\pi + (-1 + k) \theta)) / k) \\
s_{\theta,k}' &= \sin((2  (\pi + (-1 + k) \theta)) / k), \\
c_{\theta,k} &= \cos((2  (\pi - \theta)) / k), \\
s_{\theta,k} &= \sin((2  (\pi - \theta)) / k).
\end{align*}
We are now ready to write down the promised NLP, which, based on the observations above, is defined on the variables $t \in \mathbb{R}^k$ and $\theta \in \mathbb{R}$. The value of $\epsilon > 0$ is provided as part of the input and is chosen to be sufficiently small, ensuring that the expressions are well-defined over the domain of $\theta$.

\begin{mdframed}
\begin{align*}
\min_{t\in \reals^k, \theta\in \reals} & ~~
\frac{1}{c}
\log\left( \frac{1+\sin(\theta)}{1-\sin(\theta)}\right) 
+
\left(1+ \frac{1}{k} \right)
\left(1- \frac{2\theta}{c} \right)
\left(
\|OA_0\| + \frac{1}{k+1}\sum_{i=0}^{k-1} (k-i) \|A_iA_{i+1}\|
\right)
			\tag{NLP-avg($k,c,\epsilon$)}	\label{NLP-avg} \\
\text{s.t.} & ~~ \sin\left(\frac{c-2\theta}{k}\right) t_i \geq 1-\cos\left(\frac{c-2\theta}{k}\right) , ~~i=1,\ldots,k, \\
& ~~ 0 \leq \theta \leq  \min\left\{\frac{(1-\epsilon)c}{2},\frac{\pi}{2}-\epsilon\right\}.
\end{align*}
\end{mdframed}

\begin{lemma}
\label{lem: NLP feasible to p(c) and cost is correct}
Fix $\epsilon>0$, small enough. 
For any feasible solution $(\theta,t)\in \reals^k$ to~\eqref{NLP-avg} with $k\geq 5$,
\textsc{EPS}$(\theta,t_1,\ldots,t_k)$ is a feasible inspection algorithm to \pshoreline{c}. 
Moreover the objective to the NLP quantifies the corresponding average case cost of the algorithm. 
\end{lemma}

\begin{proof}
For a solution $(\theta,t)\in \reals^k$ to~\eqref{NLP-avg}, we observe that 
$$t_i \geq (1-\coss{\frac{c-2\theta}k})/\sinn{\frac{c-2\theta}k},$$
for $i=1,\ldots,k$. 
By Lemma~\ref{lem: polysegment performance for continuous avg}
the trajectory \textsc{PS}$(t_1,\ldots,t_k)$ is feasible to \ashoreline{\infty}{\theta}{c}, and if its average case cost to \ashoreline{k}{\theta}{c} equals $s$, then the average case cost to \ashoreline{\infty}{\theta}{c} is at most $\left(1+\frac{1}{k}\right) s$. 

Then, by Lemma~\ref{lem: reduction from discrete to continuous}, we obtain that \textsc{EPS}$(\theta,t_1,\ldots,t_k)$ is also feasible to \pshoreline{c}, with overall cost 
$\frac{1}{c} 
\log \left( \frac{1+\sinn{\theta}}{1-\sinn{\theta}} \right)
+
\left(1+\frac{1}{k}\right)\left(1 - \frac{2\theta}{c} \right) s
$.
Finally, Lemma~\ref{lem: polysegment performance for discrete} computes quantity $s$, which results to the objective of~\eqref{NLP-avg}, as promised. 
\qed  \end{proof}

We are ready to provide upper bounds to the average inspection cost for \pshoreline{c}. 
\begin{theorem}
\label{thm: upper bounds to average inspection partial}
\pshoreline{c} admits average inspection cost solution as depicted in Figure~\ref{fig: partialAVGperformance}.
\end{theorem}

\begin{proof}[of Theorem~\ref{thm: upper bounds to average inspection partial}]
For $c$ from $0.01$ up to $2\pi$ and with step size $0.01$, we compute a numerical solution to \eqref{NLP-avg} with $k = 1000$ and $\epsilon = 10^{-3}$, and we apply Lemma~\ref{lem: NLP feasible to p(c) and cost is correct}. The numerical solutions were obtained using Julia's JuMP and Ipopt~\cite{Ipopt,dunning2017jump}, which use an interior point method for solving NLPs and terminate with a certificate of local optimality. Hence, even though a general proof of global optimality eludes us, the reported results are provably locally optimal.

It is important to note that all upper bounds were obtained by trajectories \textsc{EPS}$(\theta,t_1,\ldots,t_k)$, which were originally designed for the discrete problem of inspecting $\Theta(k)$ disk points. Figure~\ref{fig: worstcaseOptDeploymentAngle} on page~\pageref{fig: worstcaseOptDeploymentAngle}shows the reported (minimizer) deployment angle $\theta=\theta(c)$, which exhibits discontinuity around $c\approx 2.39$. For all values of the reported optimizer $\theta(c)$, we see that the constraint $\theta \leq \min\{(1-\epsilon)c/2,\pi/2-\epsilon\}$ is satisfied with slack for $c\geq 2.39$, and hence the value of $\epsilon$ does not affect the optimizer (also $\epsilon$ was introduced only for technical reasons in order to be able to solve the NLP) for these values of $c$. For $c\leq 2.39$, the constraint $\theta \leq (1-\epsilon)c/2$ becomes tight, and then we report the abstract average inspection cost 
$$
\frac{1}{c} \log \left( \frac{1+\sinn{c/2}}{1-\sinn{c/2}} \right),
$$ 
which we obtain by using $\theta = c/2$ in Lemma~\ref{lem: reduction from discrete to continuous}, part (c). The latter corresponds to the average inspection cost of algorithm \textsc{EPS}$(c/2,t_1,\ldots,t_k)$, where the algorithm's trajectory simplifies to $O\rightarrow A=(1,\tann{c/2})$, and hence the value of $k$ is irrelevant.

Notably, we also observe a discontinuity in the reported average inspection cost at $c=2.39$ (by an additive term of less that $10^{-3}$). This is because the analysis for $c \geq 2.39$ is formal and abstract, whereas for smaller values of $c$, we rely on numerical computations, in which explicitly, we pay an additional factor of $1+1/k$ in the objective as mandated by Lemma~\ref{lem: polysegment performance for continuous avg}, and which is the result of converting a trajectory that was originally designed for the discrete inspection problem to the continuous inspection problem. 

Finally, it is interesting to observe that there is no value of $c$ for which the optimizers $t_i=t_i(c)$ are close to being tight to the required lower bound of $\left(1-\coss{\frac{c-2\theta}k}\right)/\sinn{\frac{c-2\theta}k}=O(1/k)$. Indeed, $\min_i t_i(c)$ is decreasing in $c$, and for $c=2\pi$, we find that $\min_i t_i \geq 0.2$. 

We display some optimizers $t\in \reals^k$ for certain values of $c$ in 
Figure~\ref{fig: tvaluesn n=1,2} on page~\pageref{fig: tvaluesn n=1,2} and in Figure~\ref{fig: t values 5,6,7pOver4} on page~\pageref{fig: t values 5,6,7pOver4}, where we see that they are bounded away from $0$. 
The corresponding inspection trajectories, which are far from touching the perimeter of the disk, are depicted in Figure~\ref{fig: avg case trajectories n=1,2} (on page~\pageref{fig: avg case trajectories n=1,2}), and in Figure~\ref{fig: trajectories 5,6,7pOver4} (on page~\pageref{fig: trajectories 5,6,7pOver4}).
\qed  \end{proof}

We conjecture that our reported average inspection times are optimal. Interestingly, from our numerical results, we see that the reported average inspection cost is eventually very close to be linear, i.e. it is well approximated by $0.5395026c+0.1637722$ for $c\geq 4.3$. 
Some indicative and notable trajectories solving \pshoreline{c} and proving Theorem~\ref{thm: upper bounds to average inspection partial} are those corresponding to $c=2\pi$ and $c=\pi$, see Figure~\ref{fig: avg case trajectories n=1,2} on page~\pageref{fig: avg case trajectories n=1,2}, and to $c=i \frac{\pi}4$ for $i=5,6,7$, see Figure~\ref{fig: trajectories 5,6,7pOver4}.

\begin{figure}[h!]
    \centering
        \includegraphics[width=7cm]{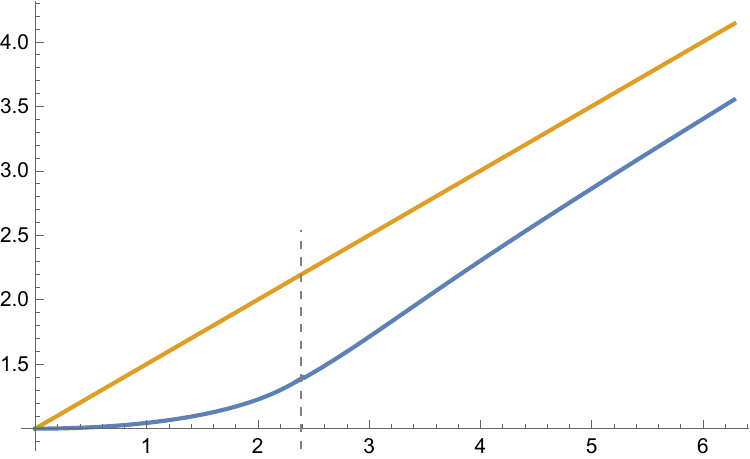}
        \caption{The reported average inspection cost (in blue) of algorithm \textsc{EPS}$(\theta,t_1,\ldots,t_k)$ as per Theorem~\ref{thm: upper bounds to average inspection  partial}. For values $c\leq 2.39$, our solution admits an analytic formula for the average cost. 
        For comparison, we also display in yellow the performance $1+c/2$ of the plain-vanilla randomized algorithm.}
        \label{fig: partialAVGperformance}
\end{figure}

\begin{figure}[h!]
    \centering
    \begin{subfigure}[t]{0.32\textwidth}
        \centering
        \includegraphics[width=5cm]{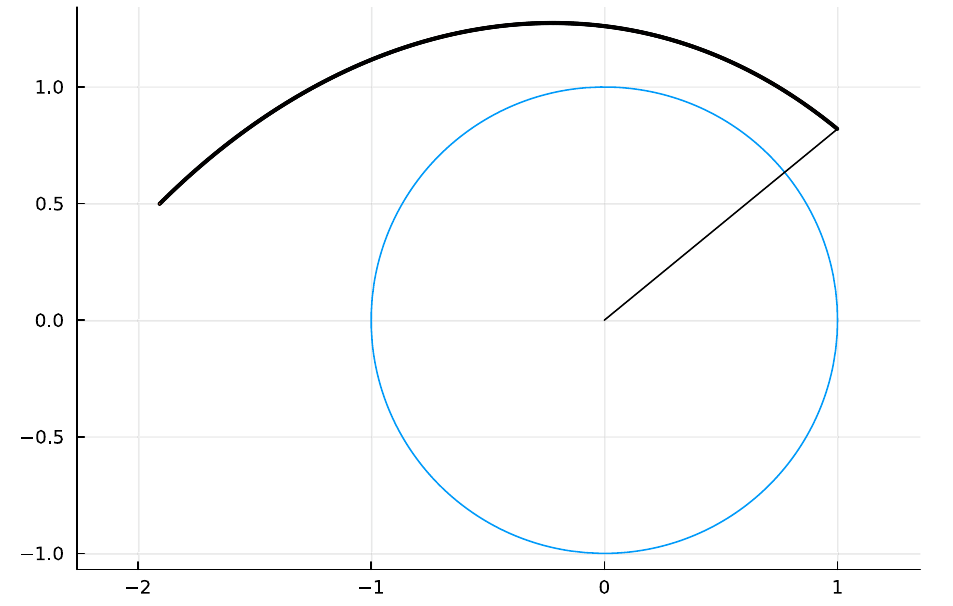}
        \caption{The trajectory \textsc{EPS}$(\theta,t_1,\ldots,t_k)$ solving \pshoreline{5\pi/4}. The corresponding $t_i$ values are seen in Figure~\ref{fig: 5pOver4-t}.}
        \label{fig: 5pOver4}
    \end{subfigure}
    \hfill
    \begin{subfigure}[t]{0.32\textwidth}
        \centering
        \includegraphics[width=4.8cm]{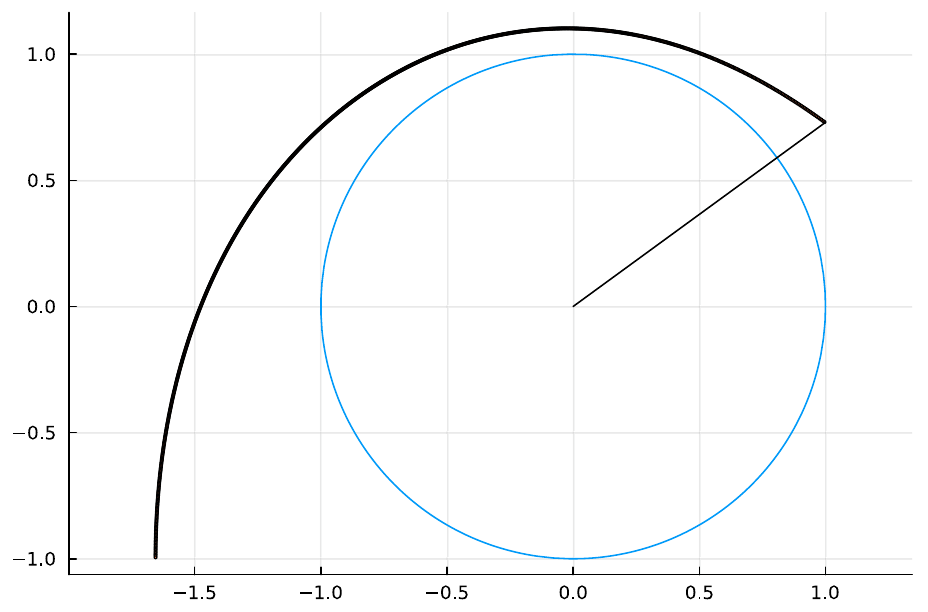}
        \caption{The trajectory \textsc{EPS}$(\theta,t_1,\ldots,t_k)$ solving \pshoreline{3\pi/2}. The corresponding $t_i$ values are seen in Figure~\ref{fig: 6pOver4-t}.}  
\label{fig: 6pOver4}
    \end{subfigure}
    \hfill
    \begin{subfigure}[t]{0.32\textwidth}
        \centering
        \includegraphics[width=4.8cm]{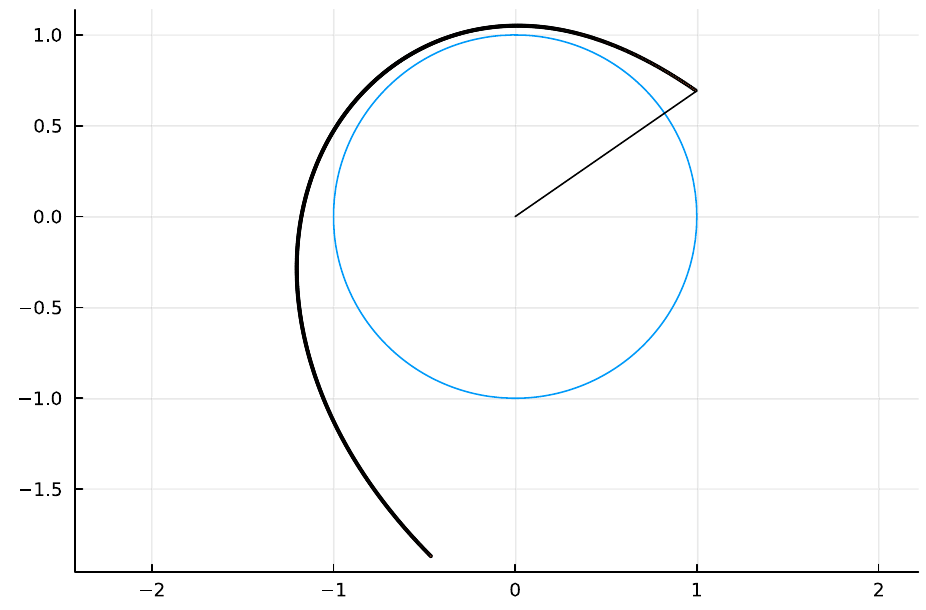}
        \caption{The trajectory \textsc{EPS}$(\theta,t_1,\ldots,t_k)$ solving \pshoreline{7\pi/4}. The corresponding $t_i$ values are seen in Figure~\ref{fig: 7pOver4-t}.}  
\label{fig: 6pOver4}
\label{fig: 7pOver4}
    \end{subfigure}
    \caption{Trajectories for minimizing the average inspection cost for \pshoreline{c}, $c=i \frac{\pi}4$, where $i=5,6,7$.     }
    \label{fig: trajectories 5,6,7pOver4}
\end{figure}


We are now ready to prove Theorem~\ref{thm: general avg upper bound}.

\begin{proof}[of Theorem~\ref{thm: general avg upper bound}]
The upper bounds on the average inspection cost for \shoreline{n} are obtained from the upper bounds of \pshoreline{c} in Theorem~\ref{thm: upper bounds to average inspection partial}, by setting $c=2\pi/n$. Since all agents are identical, each is assigned an arc of length $2\pi/n$ to inspect, and the $n$ agents' trajectories are obtained by rotating one agent's trajectory by an angle of $2\pi/n$, $n-1$ times.

We note that for $n\geq 3$ agents, the length $c=c(n)$ of the arc that each agent inspects satisfies $c\leq 2.39$, and as such the inspection trajectories we report have average cost 
$
\frac{1}{c} \log \left( \frac{1+\sinn{c/2}}{1-\sinn{c/2}} \right),
$
as per Lemma~\ref{lem: reduction from discrete to continuous}, part (c). More specifically, one agent's trajectory is given by the line segment $OA$, where $A=(1,\tann{c/2})$, and all other trajectories are rotations of this segment by $2\pi/n$.

It follows that only for $n=1,2$ are the trajectories more elaborate, requiring the numerical solution to \eqref{NLP-avg} for $c=2\pi$ and $c=\pi$, and for large enough $k$. In that direction, we refine the results reported in Theorem~\ref{thm: upper bounds to average inspection partial}, and we solve \eqref{NLP-avg} for $k=2000$ inspection points to obtain slightly improved upper bounds, as shown in Table~\ref{thm: general avg upper bound} on page~\pageref{thm: general avg upper bound}.

In Figure~\ref{fig: avg case trajectories n=1,2} on page~\pageref{fig: avg case trajectories n=1,2}, we depict the derived trajectory for \pshoreline{2\pi} and \pshoreline{\pi}, corresponding to solutions to \shoreline{1} and \shoreline{2}, respectively. In Figure~\ref{fig: tvaluesn n=1,2}, we show the corresponding values of $t_i$ returned by~\eqref{NLP-avg} which are bounded away from $0$. Finally, the values returned by the NLP for the angle $\theta$ for \shoreline{1} and \shoreline{2} were $\theta=0.5910554, 0.8054878$, respectively (while for $n\geq 3$, we use $\theta=\pi/n$ as explained before).
\qed  \end{proof}

\section{Worst/Average Case Inspection Trade-offs to \shoreline{1}}
\label{sec: tradeoffs}

The purpose of this section is twofold. First, we demonstrate that the upper bound we achieve for the average inspection cost of \shoreline{1} using \textsc{EPS}$(\theta, t_1, \ldots, t_k)$ represents a significant improvement over previously known trajectories. Second, we show that our techniques also address the multi-objective optimization problem of simultaneously minimizing the worst-case and average-case inspection times. This is achieved by considering the trade-off problem~\tradeoff{\lambda}, leading to the proof of one of our main contributions, Theorem~\ref{thm: pareto bounds}, which establishes Pareto upper bounds for the multi-objective inspection problem.

\subsection{Pareto Upper Bounds}
\label{sec: pareto upper bounds}

We begin by considering the \emph{Isbell-type} family of trajectories for \shoreline{1}, introduced in~\cite{isbell1957optimal} as part of Theorem~\ref{thm: original Isbell}, parameterized by the deployment angle $\theta$. Figure~\ref{fig: convexcurve2} shows this trajectory for \pshoreline{c}. This family of trajectories also appears in Theorem~\ref{thm: optimal worst case partial inspection sol}, where it establishes the optimal worst-case inspection solution for \pshoreline{c} with $c \in [5\pi/6, 2\pi]$ (where \pshoreline{2\pi} corresponds to \shoreline{1}).

In the \emph{Isbell-type} trajectory with parameter $\theta \in [0, \pi/2)$, we define the point $A_0 = (1, \tann{\theta})$. The agent starts from the center of the disk and moves to $A_0$ (see Figure~\ref{fig: convexcurve2} on page~\pageref{fig: convexcurve2} and Figure~\ref{fig: trajectoryl1} on page~\pageref{fig: trajectoryl1}). The agent then follows the tangent point $D' = \mathcal P_{2\theta}$ on the disk, searches the disk counterclockwise to point $D = \mathcal P_{3\pi/2}$, and finally moves along the projection to the tangent line at $\mathcal P_{2\pi}$, given by the line $x = 1$.

\begin{lemma}
\label{lem: isbell performance}
For every $\theta \in [0,\pi/2)$, the Isbell-type trajectory with deployment angle $\theta$ has the following inspection performance. 
\begin{enumerate}[label=(\alph*)]
\item The worst case inspection equals $I_W(\theta) := 1/\coss{\theta}+\tann{\theta}+3\pi/2-2\theta+1$.
\item The average case inspection cost equals 
$$
I_A(\theta):=
\frac1{2\pi}
\left( 
\log \left( \frac{1+\sinn{\theta}}{1-\sinn{\theta}} \right)
+
2 \theta ^2-4 \pi  \theta +2 (\pi -\theta ) \tan (\theta )+\frac{2 (\pi -\theta )}{\coss{\theta}}+\frac{15 \pi ^2}{8}+
\log 2
\right).
$$
\end{enumerate}
\end{lemma}

\begin{proof}
\begin{enumerate}[label=(\alph*)]
\item In Theorem~\ref{thm: general wrs optimal}, we analyzed the performance of the generalized Isbell-type algorithm for solving \pshoreline{c} for all $c \in [\pi/2, 2\pi]$. In our case, $c = 2\pi$ (corresponding to \shoreline{1}), resulting in the worst-case cost denoted as $f_2(\theta)$ in the proof of Theorem~\ref{thm: general wrs optimal}, which is the worst case inspection cost promised in the current lemma. 

\item The points $\mathcal P_\phi$, with $\phi \in [0,2\pi]$, are inspected in 3 phases by the Isbell-type trajectory with parameter $\theta$. 
First, the portion $O \rightarrow A_0$ of the trajectory inspects, by Lemma~\ref{lem: inspection range of first point}, all points with $\phi \in [0,2\theta]$. The cumulative cost of the points inspected in this phase is, by Lemma~\ref{lem: reduction from discrete to continuous avg} equal to $\log \left( \frac{1+\sinn{\theta}}{1-\sinn{\theta}} \right)$. 

Then, the movement $A_0 \rightarrow D'$ inspects only the point $\mathcal P_{2\theta}$ already inspected in the previous phase. This is followed by the counter clockwise movement along the arc with endpoints $D'=\mathcal P_{2\theta}$ and $D=\mathcal P_{3\pi/2}$. For each $\phi \in [2\theta,3\pi/2]$, we have $\inspect{\mathcal P_{\phi}} = OA_0+A_0D'+\phi -2\theta$, where $OA_0 = 1/\coss{\theta}$ as calculated in the proof of Lemma~\ref{lem: reduction from discrete to continuous wrs}, and $A_0D = A_0\mathcal P_0 = \tann{\theta}$, by the definition of point $A_0$ (see also Figure~\ref{fig: convexcurve2}). We conclude that the cumulative inspection cost in this phase is
\begin{align*}
\int_{2\theta}^{3\pi/2} \inspect{\phi} \dd \phi 
& = 
\int_{2\theta}^{3\pi/2}\left(\tfrac{1}{\coss{\theta}}+\tann{\theta}+\phi -2\theta\right)\dd \phi \\
& = \left(\tfrac{3\pi}{2}-2\theta\right)\left(\tfrac{1}{\coss{\theta}}+\tann{\theta} -2\theta\right)+ \tfrac{9 \pi ^2}{8}-2 \theta ^2
\end{align*}
 
Finally, the trajectory concludes with movement $\mathcal P_{3\pi/2} \rightarrow F$, where $F$ is the intersection of point of tangent lines $\mathcal L_{3\pi/2}$ and $\mathcal L_{0}=\mathcal L_{2\pi}$, that is $F=(-1,0)$, see Figure~\ref{fig: finalpointavg}. This phase inspects $\mathcal P_\phi$ with $\phi \in [3\pi/2,2\pi]$ and for each such point we find its inspection time $\inspect{\phi}$.   
 \begin{figure}[h!]
    \centering
        \includegraphics[width=5.0cm]{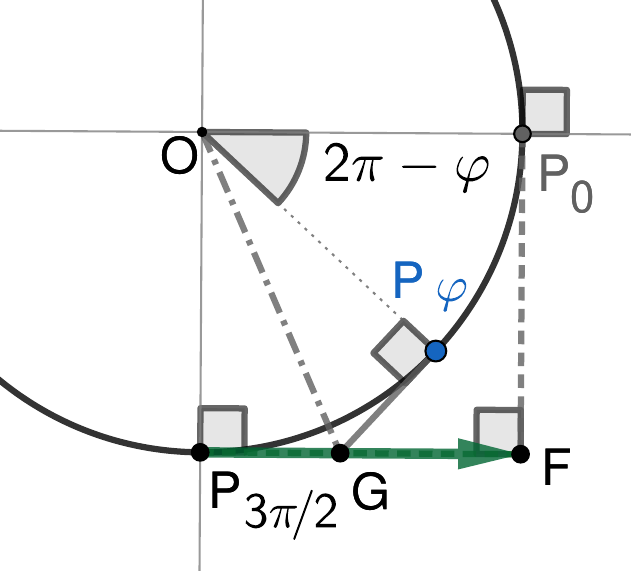}
        \caption{The last phase of the Isbell-type algorithm for \pshoreline{2\pi} (also \shoreline{1}). 
        }
        \label{fig: finalpointavg}
\end{figure}
Note that in order to reach point $\mathcal P_{3\pi/2}$, and by our previous analysis, time $1/\coss{\theta}+\tann{\theta}+3\pi/2 -2\theta$ has passed. The additional time that is needed to inspect $\mathcal P_\phi$ is by Lemma~\ref{lem: inspection on tangent line} equal to the length of $\mathcal P_{3\pi/2}G$. Now, the angle of triangle $GO\mathcal P_{3\pi/2}$ (at vertex $O$) is $(\pi/2-(2\pi-\phi))/2 = \phi/2-3\pi/4$, since triangles $GO\mathcal P_{3\pi/2}, GO\mathcal P_\phi$ are equal. But then $\mathcal P_{3\pi/2}G = \tann{\phi/2-3\pi/4}$, making the cumulative cost in this case equal to
\begin{align*}
\int_{3\pi/2}^{2\pi} \inspect{\phi} \dd \phi 
& = 
\tfrac{\pi}2
\left( \tfrac{1}{\coss{\theta}}+\tann{\theta}+3\pi/2 -2\theta \right)
+
\int_{3\pi/2}^{2\pi} \tann{\phi/2-3\pi/4} \dd \phi \\
& =
\tfrac{\pi}2
\left(\tfrac{1}{\coss{\theta}}+\tann{\theta}+3\pi/2 -2\theta\right)
+
\log 2
,
\end{align*}
where the last equality is based on that $\int \tann{x} = - \log(\coss{x})$. 
Adding then all cumulative inspection costs and scaling by $2\pi$ gives the average inspection cost, as described in the lemma. 
\end{enumerate}
\qed  \end{proof}

The value of $\theta$ that minimizes the worst-case inspection cost $I_W(\theta)$, as described in Lemma~\ref{lem: isbell performance}, is $\theta_W = \pi/6 \approx 0.523599$ (see also Theorem~\ref{thm: original Isbell} from \cite{isbell1957optimal}, or our more general Theorem~\ref{thm: general wrs optimal} for $n = 1$). For this value of $\theta_W$, we have
$I_W(\theta_W)  = 1 + \sqrt{3} + \frac{7\pi}{6} \approx 6.39724$ 
and
$I_A(\theta_W)  = \frac{5}{2\sqrt{3}} + \frac{91\pi}{144} + \frac{\log(6)}{2\pi} \approx 3.71386$.
Similarly, the value of $\theta$ that minimizes the average inspection cost $I_A(\theta)$, from Lemma~\ref{lem: isbell performance}, is $\theta_A = 0.592334$, which results in :
\begin{align*}
I_W(\theta_A) & \approx 6.406, \\
I_A(\theta_A) & \approx 3.70737.
\end{align*}
We expect the Pareto frontier of Isbell-type algorithms to be determined by $\theta \in [0.52359, 0.59233]$. However, for completeness, we plot the pairs $I(\theta) := (I_W(\theta), I_A(\theta))$ for all $\theta \in [0.4, 0.75]$ in Figure~\ref{fig: all tradeoff}, on page~\pageref{fig: all tradeoff}.

In Theorem~\ref{thm: general avg upper bound}, we achieved an average inspection cost as low as $3.5509015$. We are ready to demonstrate how to significantly improve both the worst-case and average inspection costs across the entire Pareto frontier, effectively proving Theorem~\ref{thm: pareto bounds}.

\begin{proof}[of Theorem~\ref{thm: pareto bounds}]
The upper bound to the Pareto frontier will be obtained by the family of algorithms \textsc{EPS}$(\theta,t_1,\ldots,t_k)$, with $k=1000$, and for appropriate choices of parameters $\theta,t_i$. 

For this, we recall Lemma~\ref{lem: reduction from discrete to continuous wrs}, from which for any feasible choice of the parameters, we obtain a formula $S_W(\theta,t_1,\ldots,t_k)$ for the worst case cost. 
Let also $S_A(\theta,t_1,\ldots,t_k)$ denote the objective of \eqref{NLP-avg}, with $c=2\pi$ and $\epsilon =10^{-3}$. For all NLP solutions we obtain, the optimizer values for $\theta$ are bounded away from $\pi/2$ (in fact, they admit values in $[0.5,0.6]$), and hence the value of $\epsilon$ does not affect the cost of the solution, rather only is there to justify its correctness (see also Figure~\ref{fig: theta valus for trade-offs}).
\begin{figure}[h!]
    \centering
    \includegraphics[width=7.5cm]{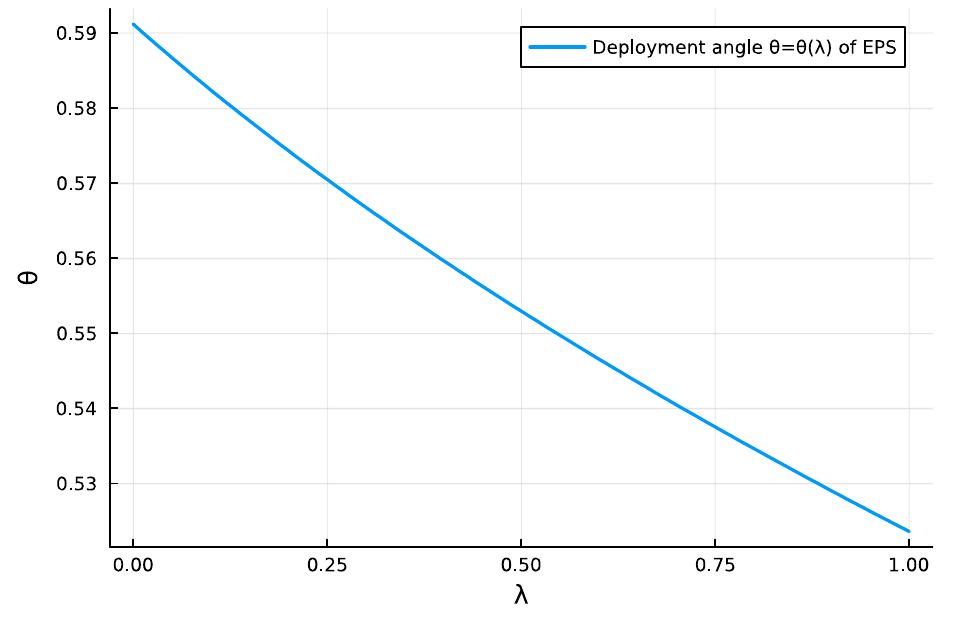}
    \caption{
The value of the initial deployment angle $\theta$ of trajectory \textsc{EPS}$(\theta,t_1,\ldots,t_k)$ designed to minimize the trade-off problem \tradeoff{\lambda}, as a function of $\lambda \in [0,1]$. 
}
\label{fig: theta valus for trade-offs}
\end{figure}
According to Lemma~\ref{lem: NLP feasible to p(c) and cost is correct}, $S_A(\theta,t_1,\ldots,t_k)$ correctly quantifies an upper bound to the average inspection cost to \textsc{EPS}$(\theta,t_1,\ldots,t_k)$, for any feasible solution $(\theta,t_1,\ldots,t_k)$ to the NLP. 

Consequently, the new nonlinear program with objective 
$$ \lambda S_W(\theta, t_1, \ldots, t_k) + (1-\lambda) S_A(\theta, t_1, \ldots, t_k) $$ and subject to the same constraints as in \eqref{NLP-avg}, correctly quantifies the objective for the trade-off problem \tradeoff{\lambda}.

For every $\lambda$ from $0$ to $1$ and with step size $0.01$, solving the new NLP we obtain a solution $(\theta(\lambda),t_1(\lambda),\ldots,t_k(\lambda))$, and we define quantities 
$
N_W(\lambda) := S_W(\theta(\lambda),t_1(\lambda),\ldots,t_k(\lambda)),
$
and
$N_A(\lambda) := S_A(\theta(\lambda),t_1(\lambda),\ldots,t_k(\lambda)).
$
Thus, $N_W(\lambda)$ and $N_A(\lambda)$ accurately quantify the worst-case inspection cost and average inspection cost, respectively, for feasible \textsc{EPS}$(\theta(\lambda), t_1(\lambda), \ldots, t_k(\lambda))$. As expected, $\lambda = 1$ provides an approximation of the worst-case optimal inspection trajectory and a precise quantification of its average inspection cost.
Indeed, $N_W(1) \approx 6.3972578$ and $N_A(1) \approx 3.71130923$ (compare these to $I_W(\theta_W)$ and $I_A(\theta_W)$ respectively of the Isbell-type algorithm, and note that due to the feasibility lower bound on the $t_i$'s imposed in the NLP, the average inspection cost is slightly improved, whereas the worst case cost is slightly compromised). 
Similarly, the solution for $\lambda = 0$ yields the average inspection upper bound $N_A(0)$ from Theorem~\ref{thm: general avg upper bound} for \pshoreline{2\pi} (equivalently, \shoreline{1}) and a worst-case inspection cost of $N_W(0) = 6.8867382$.
We conclude with Figure~\ref{fig: all tradeoff} on page~\pageref{fig: all tradeoff}, where we depict the pair $(N_W(\lambda),N_A(\lambda))$ for $\lambda \in [0,1]$. 
\qed  \end{proof}

Trajectories that certify the reported average inspection cost upper bounds for~\tradeoff{\lambda} are depicted in Figure~\ref{fig: trajectoriesl}, with $\lambda = 1$ corresponding to optimizing the worst-case cost, and for $\lambda = 2/3$ and $\lambda = 1/3$. Recall that the trajectory certifying the upper bound for~\tradeoff{0}, which minimizes the average inspection cost, is shown in Figure~\ref{fig: avg trajectory n=1} on page~\pageref{fig: avg trajectory n=1}.
The values of $t \in \mathbb{R}^k$ obtained as solutions to the nonlinear program described in the proof of Theorem~\ref{thm: pareto bounds} are shown in Figure~\ref{fig: trajectoriesl-t} on page~\pageref{fig: trajectoriesl-t} for~\tradeoff{1}, ~\tradeoff{2/3}, and ~\tradeoff{1/3}. Additionally, Figure~\ref{fig: tvaluesn1} on page~\pageref{fig: tvaluesn1} presents the values of these parameters for~\tradeoff{0}.

\begin{figure}[h!]
    \centering
    \begin{subfigure}[t]{0.32\textwidth}
        \centering
        \includegraphics[width=5cm]{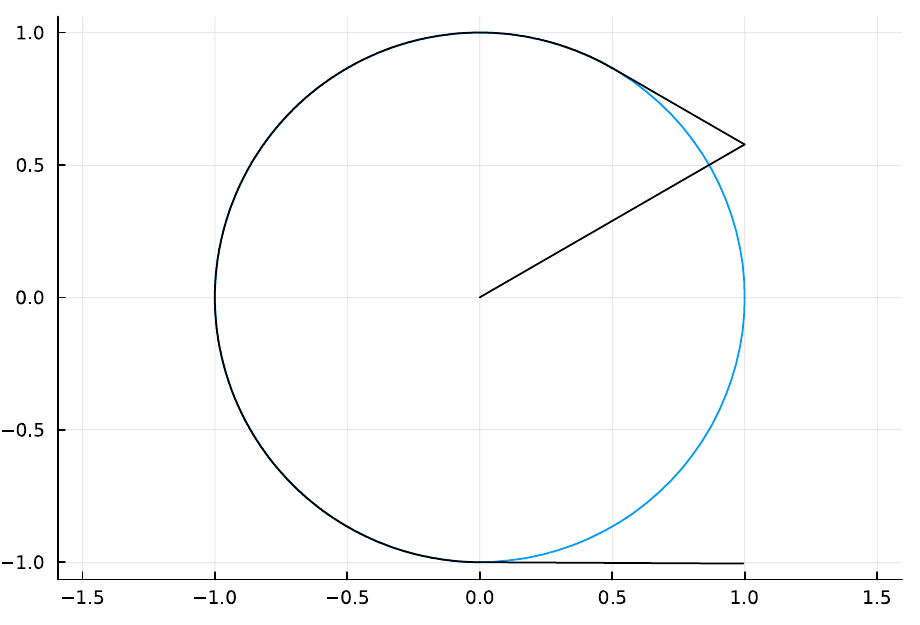}
        \caption{The trajectory of the optimal \textsc{EPS}$(\theta,t_1,\ldots,t_k)$ for solving \tradeoff{1}.
       The values of the corresponding values $t_i$ are see in Figure~\ref{fig: trajectoryl1-t}.
}
        \label{fig: trajectoryl1}
    \end{subfigure}
    \hfill
    \begin{subfigure}[t]{0.32\textwidth}
        \centering
        \includegraphics[width=4.8cm]{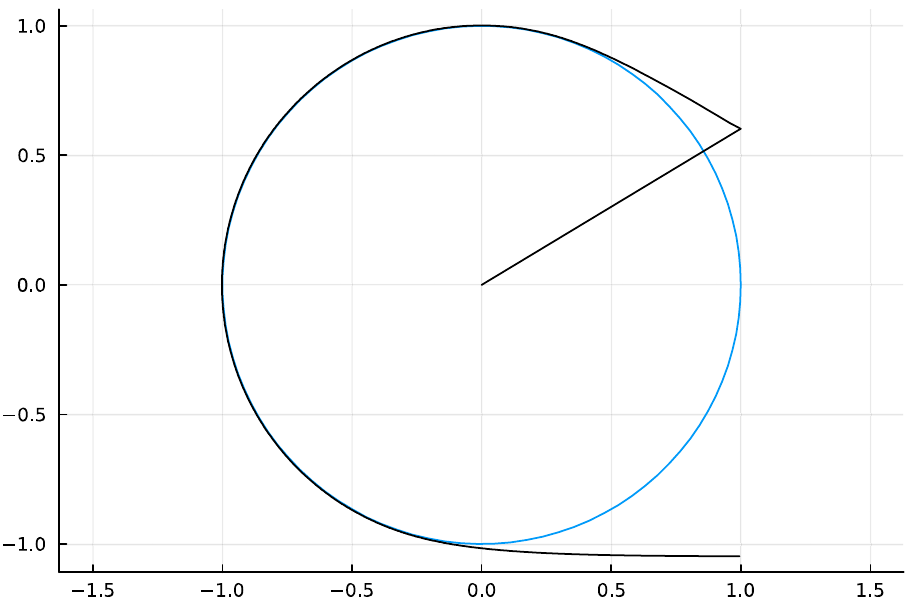}
        \caption{The trajectory of the optimal \textsc{EPS}$(\theta,t_1,\ldots,t_k)$ for solving \tradeoff{2/3}. The values of the corresponding values $t_i$ are see in Figure~\ref{fig: trajectoryl23-t}.
        }
        \label{fig: trajectoryl23}
    \end{subfigure}
    \hfill
    \begin{subfigure}[t]{0.32\textwidth}
        \centering
        \includegraphics[width=4.8cm]{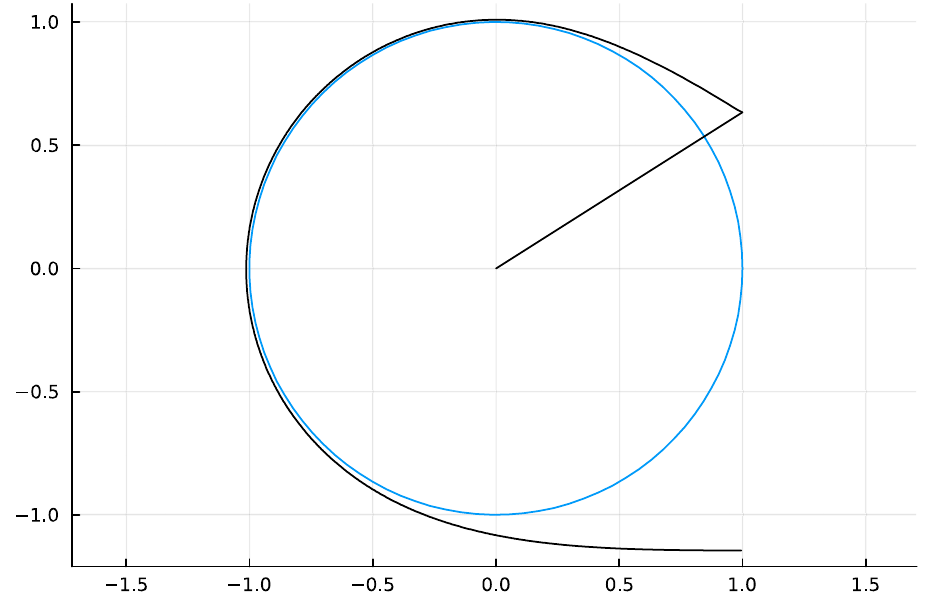}
        \caption{The trajectory of the optimal \textsc{EPS}$(\theta,t_1,\ldots,t_k)$ for solving \tradeoff{1/3}. The values of the corresponding values $t_i$ are see in Figure~\ref{fig: trajectoryl13-t}.}
        \label{fig: trajectoryl13}
    \end{subfigure}
    \caption{
   Agent's trajectories for trade-off problem \tradeoff{\lambda} of the proof of Theorem~\ref{thm: general avg upper bound}, for $\lambda=1,2/3,1/3$, using algorithm \textsc{EPS}$(\theta,t_1,\ldots,t_k)$ with $k=1000$. 
   The trajectory corresponding to $\lambda=0$ (i.e. the best average performance of \textsc{EPS}$(\theta,t_1,\ldots,t_k)$ solving \shoreline{1}) can be seen in Figure~\ref{fig: avg trajectory n=1}.
    }
    \label{fig: trajectoriesl}
\end{figure}

\section{Discussion}
\label{sec: conclusions}

Our results provide a unified framework for analyzing inspection trajectories on the disk in both worst- and average-case settings. For the worst case, we recover and extend known optimal solutions by establishing a general theorem for partial inspection problems, thereby offering a single proof that subsumes previous results. For the average case, we revisit the heuristic analysis of Gluss~\cite{gluss1961alternative}, identify numerical errors in the reported values, and present corrected bounds. Our systematic approach based on discretization and nonlinear programming yields upper bounds that improve on these corrected values, and it applies naturally to the multi-objective setting, where we derive Pareto upper bounds that balance worst- and average-case costs.

A natural limitation of our approach was the absence of a proof of global optimality for the nonlinear programs. However, subsequent work has shown that the methodology we develop here, after substantial additional technical work, does yield the true optimum of the average-case Disk-Inspection problem~\cite{georgiou2025optimalaveragediskinspectionfermats}. The optimal cost was proved to be $3.549259$, only slightly below our best reported upper bound of $3.5509015$. The proof relies crucially on the nonlinear programming framework introduced here, but most importantly it reduces the problem to a one-parameter family of trajectories using principles from optics and Fermat’s Principle of Least Time. This reduction leads to an ODE system whose solution defines the provably optimal trajectory, which, as in our numerical constructions, avoids touching the disk (contrary to the conjecture of Gluss~\cite{gluss1961alternative}).

Looking ahead, our framework suggests several promising directions. One is to study trade-offs in more general geometric environments, or in higher dimensions, where inspection or patrolling tasks arise in practical applications. Another is to adapt these techniques to related problems in mobile agent computing, such as evacuation or rendezvous, where both worst- and average-case guarantees are of interest.

\bibliographystyle{plain}
\bibliography{BiblioSearch-new}

\appendix

\section{Omitted Proofs}
\label{sec: omitted proofs}

\subsection{Proofs Omitted from Section~\ref{sec: solve inspection by partial inspection}}
\label{sec: proofs from sec solve inspection by partial inspection}

\begin{lemma}
\label{lem: minimizer of f1}
For every $c\in [0,3\pi/2]$, function $f_1(\theta) := 1/\coss{\theta}+\sinn{c-\theta-\pi/2}/\coss{\theta}+1$, subject to that $\theta \geq c/2-\pi/4$ and $\theta<\pi/2$, attains the following minima:
$$
\begin{cases}
f_1(c/2), 			& \textrm{ if}~ 0 \leq c \leq 2\pi/3\\ 
f_1(\pi-c),	 	& \textrm{ if}~ 2\pi/3 < c \leq 5\pi/6\\ 
f_1(c/2-\pi/4),		& \textrm{ if}~ 5\pi/6 < c \leq 3\pi/2.
\end{cases}
$$
\end{lemma}

\begin{proof}
For every $c\in [0,3\pi/2]$, we aim to find the minimum of the function
$
f_1(\theta) = \frac{1 + \sin(c - \theta-\pi/2)}{\cos\theta} + 1,
$
subject to the constraints \(\theta \geq \dfrac{c}{2} - \dfrac{\pi}{4}\) and \(\theta < \dfrac{\pi}{2}\). For this we rely on the derivative $f_1'(\theta) = ( \sinn{\theta} - \sinn{\theta})/\coss{\theta}$. Next we examine cases for the range of parameter $c$. 

Case 1: If $0\leq c\leq \pi/2$, we have that that $f_1'(\theta)=0$ if and only if $\theta=c$ (for $0\leq \theta <\pi/2$), whereas $f_1'(0)=-\sinn{c} <0$. This means that the function is initially decreasing, up to $\theta =c$. However, $c/2-\pi/4\leq \theta \leq c/2\leq c$, and therefore $f_1(\theta)$ is minimized at the end of its domain, that is at $\theta = c/2$,

Case 2: If $\pi/2\leq c\leq \pi$, and since $0\leq \theta <\pi/2$, we have that that $f_1'(\theta)=0$ if and only if $\theta=\pi-c$. At the same time, $c/2-\pi/4\leq \theta \leq c/2$, hence we see that $c/2-\pi/4\leq \pi-c \leq c/2$ if and only if $2\pi/3 \leq c \leq 5\pi/6$. 
At the same time, for every $\epsilon>0$ sufficiently small, we have that $f_1'(\pi-c-\epsilon)=
( \sinn{c+\epsilon} - \sinn{c})/\coss{c+\epsilon}$ implying that $f_1$ is decreasing for $c\leq \pi-c$. We conclude that for $c\in [2\pi/3 \leq c \leq 5\pi/6]$, $f_1(\theta)$ is minimized at $\theta=\pi-c$. 
For $c\in [\pi/2,2\pi/3]$, we have that $\pi-c \geq c/2$. Hence, $f_1(\theta)$, is decreasing up to $\theta \leq c/2$, and hence the minimum in this case is attained at $\theta=c/2$. 
Finally, for the case 2, if $c\in [5\pi/6,\pi]$, we have that $\pi-c \leq c/2-\pi/4$, hence $f_1(\theta)$ is increasing when $\theta \in [c/2-\pi/4,c/2]$, inducing the minimizer $\theta=c/2-\pi/4$. 

Case 3: Finally, if $\pi\leq c\leq 3\pi/2$, we have that $\sinn{c}\leq 0$, and hence $f_1'(\theta) = ( \sinn{\theta} - \sinn{\theta})/\coss{\theta} \geq 0$. For this reason, $f_1(\theta)$ is increasing in $[c/2-\pi/4,\pi/2)$, and so the minimum is attained at $\theta=c/2-\pi/4$.
\qed  \end{proof}

\begin{lemma}
\label{lem: minimizer of f2}
For every $c\in [\pi/2,2\pi]$, function $f_2(\theta) = 1/\coss{\theta}+\tann{\theta}+c-\pi/2-2\theta+1$, subject to that $0\leq \theta \leq c/2-\pi/4$ and $\theta<\pi/2$, attains the following minima:
$$
\begin{cases}
f_2(c/2-\pi/4),		& \textrm{ if}~ \pi/2 \leq c \leq 5\pi/6 \\
f_2(\pi/6), 			& \textrm{ if}~ 5\pi/6 < c \leq 2\pi.
\end{cases}
$$
\end{lemma}

\begin{proof}
For each fixed parameter $c\in [\pi/2,2\pi]$, we study the monotonicity of $f_2(\theta) = 1/\coss{\theta}+\tann{\theta}+c-\pi/2-2\theta+1$ in the domain identified by that $0\leq \theta \leq c/2-\pi/4$ and $\theta<\pi/2$. For this we see that 
$f_2'(\theta)=(1/\coss{\theta}+\tann{\theta})/\coss{\theta} -2$, 
and that 
$f_2''(\theta)=(1/\coss{\theta}+\tann{\theta})^2/\coss{\theta}$.

The above imply that $f_2(\theta)$ is convex, and $f_2'(\theta)$ is increasing in $[0,\pi/2)$. The only root of $f_1'(\theta)$ in $[0,\pi/2)$ is at $\theta = \pi/6$, while it is also easy to see that $f_2'(0)=-1$. Therefore, $f_2(\theta_2)$ is initially decreasing, it is minimized at $\theta=\pi/6$, and then it is increasing for up to $\theta <\pi/2$. 

We recall the domain of $f_2(\theta)$, identified by that $0\leq \theta \leq c/2-\pi/4$ and $\theta<\pi/2$. Hence, $\theta=\pi/6$ is a minimizer if and only if $\pi/6 \leq c/2-\pi/4$ iff and only if $c \in [5\pi/6,2\pi]$. 
For parameters $c\in [\pi/2, 5\pi/6]$, we see that $f_2(\theta)$ is decreasing in its domain, and hence it is minimized at $\theta = c/2-\pi/4$. 
\qed  \end{proof}

\begin{lemma}
\label{lem: f1-f2 comparison}
For every $c\in [\pi/2,3\pi/2]$, let $f_1(\theta), f_2(\theta)$ be as in Lemma~\ref{lem: minimizer of f1} and Lemma~\ref{lem: minimizer of f2}. Let $\theta_1,\theta_2$ denote the corresponding minimizers. Then, 
$$
\min\{f_1(\theta_1), f_2(\theta_2) \}=
\begin{cases}
f_1(c/2) = 1/\coss{c/2}			& \textrm{ if}~ \pi/2 \leq c \leq 2\pi/3 \\
f_1(\pi-c) = 1-2\coss{c}		& \textrm{ if}~ 2\pi/3 < c \leq 5\pi/6 \\
f_2(\pi/6)= \sqrt3+c-5\pi/6+1	& \textrm{ if}~ 5\pi/6 < c \leq 2\pi 
\end{cases}
$$
\end{lemma}

\begin{proof}
For each $c\in [0,2\pi]$, let $\theta_1=\theta_1(c),\theta_2=\theta_2(c)$ denote the minimizers of $f_1(\theta), f_2(\theta)$ as in Lemma~\ref{lem: minimizer of f1} and Lemma~\ref{lem: minimizer of f2}, and for the values of $c$ that each is defined. By direct substitution we see that 
$$
f_1(\theta_1)=
\begin{cases}
1/\coss{c/2}, 			& \textrm{ if}~ 0 \leq c \leq 2\pi/3\\ 
1-2\coss{c},	 	& \textrm{ if}~ 2\pi/3 < c \leq 5\pi/6\\ 
1+1/\sinn{c/2+\pi/4}+\tann{c/2-\pi/4},		& \textrm{ if}~ 5\pi/6 < c \leq 3\pi/2.
\end{cases}
$$
and
$$
f_2(\theta_2)=
\begin{cases}
1+1/\sinn{c/2+\pi/4}+\tann{c/2-\pi/4},		& \textrm{ if}~ \pi/2 \leq c \leq 5\pi/6 \\
\sqrt3+c-5\pi/6+1, 			& \textrm{ if}~ 5\pi/6 < c \leq 2\pi.
\end{cases}.
$$
For $c\leq \pi/2$ only $f_1(\theta_1)$ is defined, while for $c\geq 3\pi/2$ only $f_2(\theta_2)$ is defined. Taking also the piece-wise definition of the minimizers, we conclude that $\min\{f_1(\theta_1),f_2(\theta_2)\}$ over $c \in [\pi/2,3\pi/2]$ should be analyzed in the following three subintervals. 

\begin{align*}
M_1:=& \min_{c \in [\pi/2,2\pi/3]}\{f_1(\theta_1),f_2(\theta_2)\}
=
\min
\left\{
\frac{1}{\coss{c/2}}, 1+\frac{1}{\sinn{c/2+\pi/4}}+\tann{c/2-\pi/4}
\right\}
\\
M_2:=&
\min_{c \in [2\pi/3, 5\pi/6]}\{f_1(\theta_1),f_2(\theta_2)\}
=
1+ \min 
\left\{
-2\coss{c},
\frac{1}{\sinn{c/2+\pi/4}}+\tann{c/2-\pi/4}
\right\}
\\
M_3:=&
\min_{c \in [5\pi/6, 3\pi/2]}\{f_1(\theta_1),f_2(\theta_2)\}
=
1+ \min
\left\{
\frac{1}{\sinn{c/2+\pi/4}}+\tann{c/2-\pi/4}
,\sqrt3+c-5\pi/6
\right\}
\end{align*}

We claim that $M_1 = \coss{c/2}$, and for this we consider function 
$$
g_1(c) =   \frac{1}{\sinn{c/2+\pi/4}}+\tann{c/2-\pi/4}+1 - \frac1{\coss{c/2}}
$$
over $c\in [\pi/2,2\pi/3]$. We will show that $g_1(c)\geq 0$. For this we see that $g_1(2\pi/3) = \left(\sqrt{2}-1\right) \left(\sqrt{3}-1\right) >0$. Hence, it suffices to show that $g_1(c)$ is decreasing. First it is easy to see that 
$$
g_1'(c) =
\frac{1}{16} \sec ^2\left(\frac{c}{2}\right) \sec ^2\left(\frac{1}{8} (2 c+\pi )\right) \left(-4 \sin \left(\frac{c}{2}\right)-2 \sin \left(c+\frac{\pi }{4}\right)+2 \cos (c)+\sqrt{2}+2\right).
$$
Hence, it suffices to show that $h_1(c) = -4 \sin \left(\frac{c}{2}\right)-2 \sin \left(c+\frac{\pi }{4}\right)+2 \cos (c)+\sqrt{2}+2 <0$ for $c\in [\pi/2,2\pi/3]$. Indeed, we have $h_1(c) = 2-2\sqrt2 <0$. Hence, it further suffices to show that $h_1(c)$ is decreasing. 
To that end we note that $\sinn{c}$ is strictly decreasing in $c\in [\pi/2,2\pi/3]$, and hence, after performing the transformation $\sinn{c} = x$, we obtain function 
$$
s_1(x) = h_1(\arcsin(x)) = -\sqrt{2-2 x^2}-\left(\sqrt{2}-2\right) x-2 \sqrt{2-2 x}+\sqrt{2}+2,
$$
over the domain $x\in [\sqrt3/2,1]$, which we need to show is increasing. For that, we calculate $s_1'(x) = \frac{2 x}{\sqrt{2-2 x^2}}+\frac{2}{\sqrt{2-2 x}}-\sqrt{2}+2$, which is well defined for all $x\in [\sqrt3/2,1)$. But then, note that for all $x\geq 0$, we have that $s_1'(x)>0$, hence indeed, $s_1(x)$ is increasing as wanted.

Next we claim that $M_2 = 1-2\coss{c}$, and for this we consider function 
$$
g_2(c) =   \frac{1}{\sinn{c/2+\pi/4}}+\tann{c/2-\pi/4}+2\coss{c}
$$
over $c\in [2\pi/3, 5\pi/6]$. We see that $g_2(5\pi/6) =0$, and hence it is enough to show that $g_2(c)\geq 0$. For this we calculate $g_2'(c) = \frac{1}{4} \sec ^2\left(\frac{1}{8} (2 c+\pi )\right)- 2 \sin (c)$, and we show that $g_2'(c)\leq 0$, for all $c\in [2\pi/3, 5\pi/6]$, implying that $g_2(c)$ is decreasing, and hence concluding what we claimed. To show that the derivative is negative, we apply transformation $x=\coss{c/4+\pi/8}$, hence defining function 
$$
s_2(x) = g_2(4\arccos(x) - \pi/2)= 2 + 1/(4 x^2) - 16 x^2 + 16 x^4
$$
with $1/2\leq x \leq \sinn{5\pi/24}\approx 0.608761$. Clearly, it suffices to show that $s_2(x) \leq 0$ in the latter interval. Applying also the transformation $2x^2=y$, we obtain  
$ s_2(\sqrt{y/2}) = 4 y^2-8 y+\frac{1}{2 y}+2 = \frac{8 y^3-16 y^2+4 y+1}{2 y}$, and hence it suffices to prove that 
$q_2(y) = 8 y^3-16 y^2+4 y+1 \leq 0$, in the domain $1/2 \leq y \leq 2 \sin ^2\left(\frac{5 \pi }{24}\right) \approx 0.741181$ obtained from the transformation. The three real roots of  $q_2(y)$ are easy to derive: $y_1=1/2, y_2=\frac{1}{4} \left(3-\sqrt{13}\right) \approx -0.151388$ and $y_3=\frac{1}{4} \left(\sqrt{13}+3\right) \approx 1.65139$. Since the leading coefficient of $q_2(y)$ is positive, it follows that $q_2(x)$ is negative for all $1/2 \leq y \leq \frac{1}{4} \left(\sqrt{13}+3\right)$. The main claim then follows by noticing that 
$2 \sin ^2\left(\frac{5 \pi }{24}\right) \leq \frac{1}{4} \left(\sqrt{13}+3\right)$. 

Finally, we claim that $M_3 = 1+\sqrt3+c-5\pi/6$, and for this we consider function 
$$
g_3(c) = \frac{1}{\sinn{c/2+\pi/4}}+\tann{c/2-\pi/4} - \sqrt3 - c + 5\pi/6 
$$
over $c\in [5\pi/6,3\pi/2]$. We need to show that $g_3(c) \geq 0$, and we see that $g_3(5\pi/6) = 0$. Therefore, it suffices to prove that $g_3'(c) = -1+1/2\cos^2(c/4+\pi/8) \geq 0$. 
Recall that $c\in [5\pi/6,3\pi/2]$, and hence $c/4+\pi/8 \in [\pi/3,\pi/2]$. But then, since $\cos(\cdot)$ is decreasing in $[\pi/3,\pi/2]$, we see that $\cos(c/4+\pi/8)\leq \coss{\pi/3} =1/2$. This implies that 
$g_3'(c) = -1+1/4\cos^2(c/4+\pi/8) \geq -1+1=0$ as wanted. 
\qed  \end{proof}

\section{Omitted Figures}
\label{sec: omitted figures}

\subsection{Figures Omitted from Section~\ref{sec: solution to avg partial and shoreline-n} }
\label{sec: Images Omitted from Section sec: solution to avg partial and shoreline-n}

\begin{figure}[h!]
    \centering
    \begin{subfigure}[t]{0.47\textwidth}
        \centering
        \includegraphics[width=5cm]{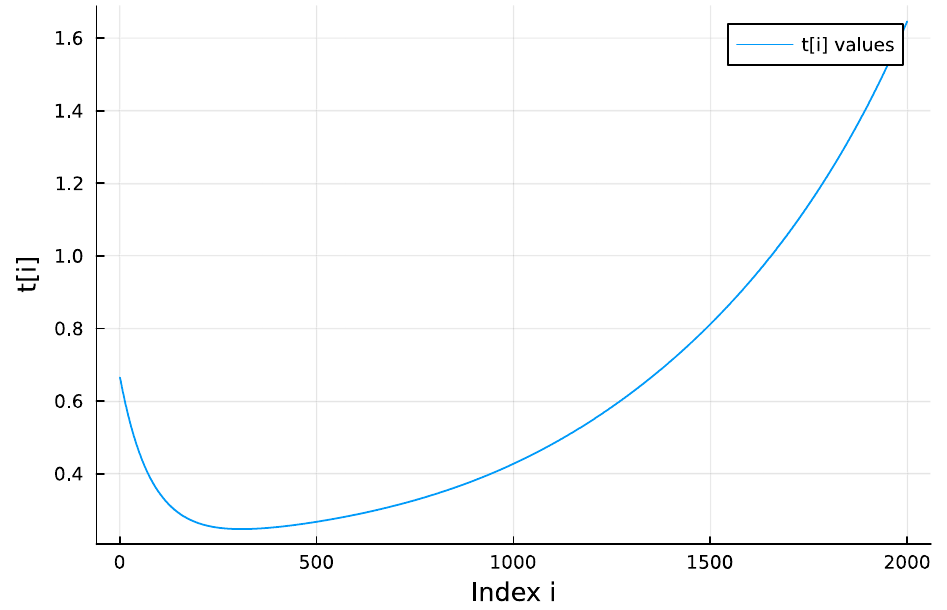}
        \caption{The values of $t_i$ for inspection trajectory \textsc{EPS}$(0.5910554,t_1,\ldots,t_k)$, with $k=2000$, for \pshoreline{2\pi} as well as for \shoreline{1}.}
        \label{fig: tvaluesn1}
    \end{subfigure}
    \hfill
    \begin{subfigure}[t]{0.47\textwidth}
        \centering
        \includegraphics[width=4.8cm]{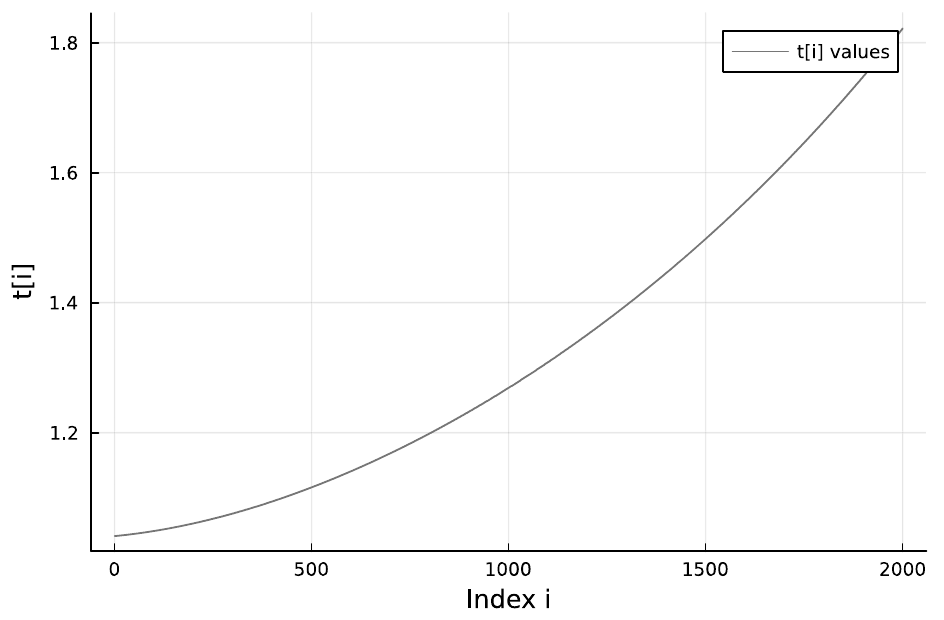}
        \caption{The values of $t_i$ for inspection trajectory \textsc{EPS}$(0.8054878,t_1,\ldots,t_k)$, with $k=2000$, for \pshoreline{\pi} as well as for \shoreline{2}.}
        \label{fig: tvaluesn1}
    \end{subfigure}
    \caption{Inspection parameter values of $t_i$ for the extended polysegment algorithm, as used in Theorem~\ref{thm: general avg upper bound}.}
    \label{fig: tvaluesn n=1,2}
\end{figure}

\begin{figure}[h!]
    \centering
    \begin{subfigure}[t]{0.32\textwidth}
        \centering
        \includegraphics[width=5cm]{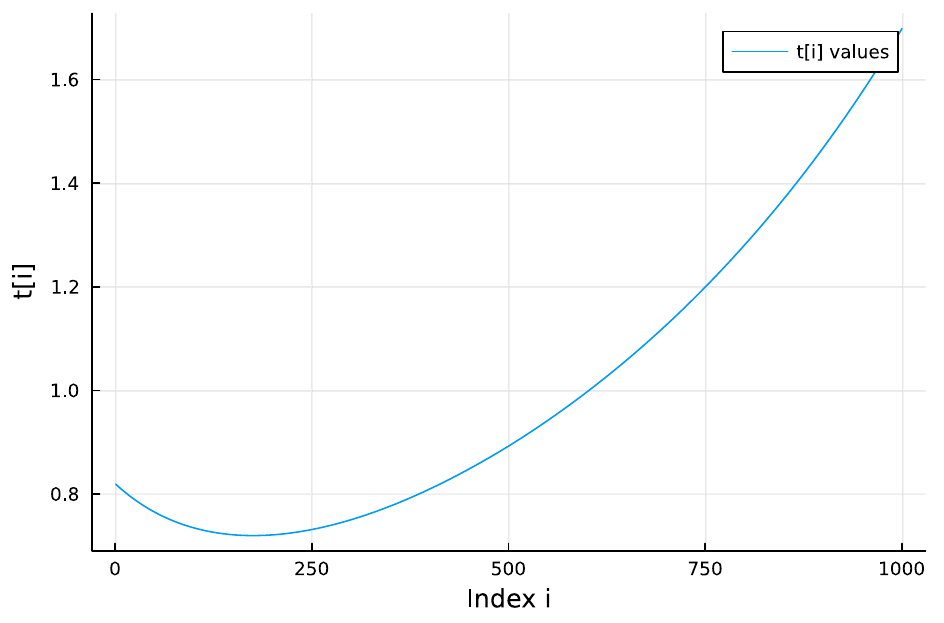}
        \caption{The values of $t_i$ of trajectory \textsc{EPS}$(\theta,t_1,\ldots,t_k)$ solving \pshoreline{5\pi/4}. The corresponding trajectory is seen in Figure~\ref{fig: 5pOver4}.
}
        \label{fig: 5pOver4-t}
    \end{subfigure}
    \hfill
    \begin{subfigure}[t]{0.32\textwidth}
        \centering
        \includegraphics[width=4.8cm]{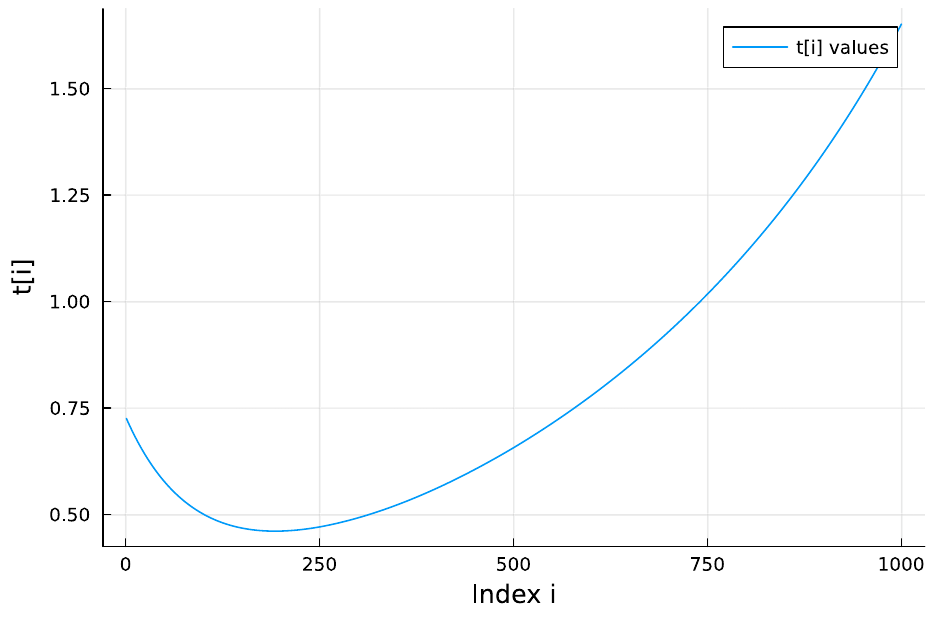}
        \caption{The values of $t_i$ of trajectory \textsc{EPS}$(\theta,t_1,\ldots,t_k)$ solving \pshoreline{6\pi/4}. The corresponding trajectory is seen in Figure~\ref{fig: 6pOver4}.
        }
        \label{fig: 6pOver4-t}
    \end{subfigure}
    \hfill
    \begin{subfigure}[t]{0.32\textwidth}
        \centering
        \includegraphics[width=4.8cm]{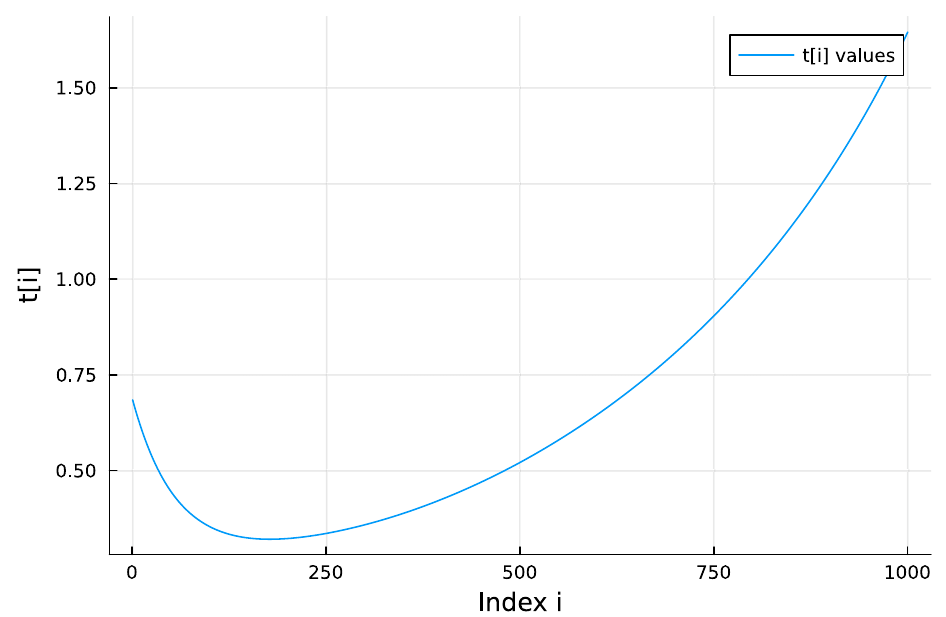}
        \caption{The values of $t_i$ of trajectory \textsc{EPS}$(\theta,t_1,\ldots,t_k)$ solving \pshoreline{7\pi/4}. The corresponding trajectory is seen in Figure~\ref{fig: 7pOver4}.
        }
        \label{fig: 7pOver4-t}
    \end{subfigure}
    \caption{Trajectories for minimizing the average inspection cost for \pshoreline{c}, $c=i \frac{\pi}4$, where $i=5,6,7$.   }
    \label{fig: t values 5,6,7pOver4}
\end{figure}

\subsection{Figures Omitted from Section~\ref{sec: tradeoffs}}
\label{sec: Figures Omitted from Section sec: tradeoffs}

\begin{figure}[h!]
    \centering
    \begin{subfigure}[t]{0.32\textwidth}
        \centering
        \includegraphics[width=5cm]{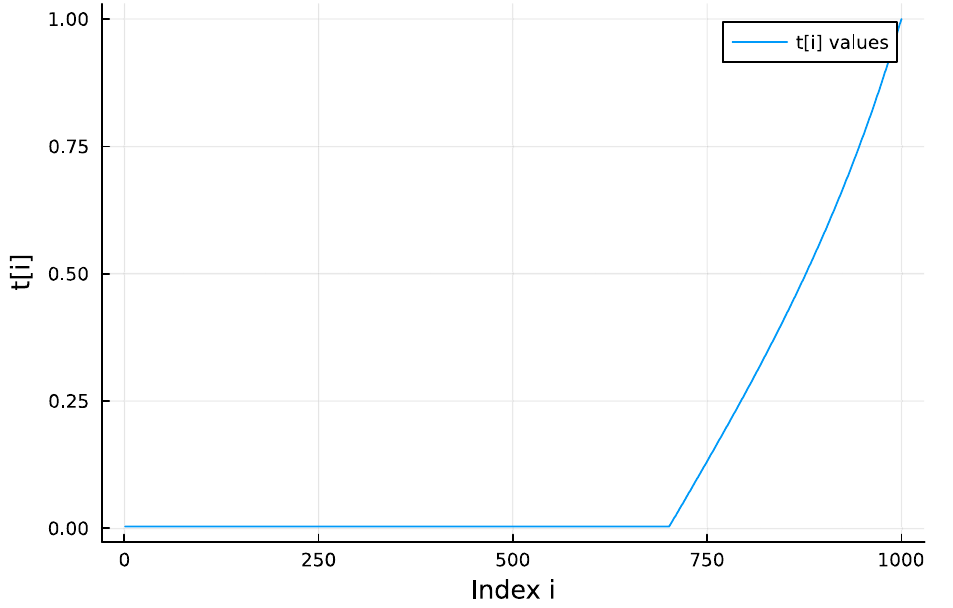}
        \caption{The values of $t_i$ of trajectory \textsc{EPS}$(\theta,t_1,\ldots,t_k)$ solving \tradeoff{1}. The corresponding trajectory is seen in Figure~\ref{fig: trajectoryl1}.
}
        \label{fig: trajectoryl1-t}
    \end{subfigure}
    \hfill
    \begin{subfigure}[t]{0.32\textwidth}
        \centering
        \includegraphics[width=4.8cm]{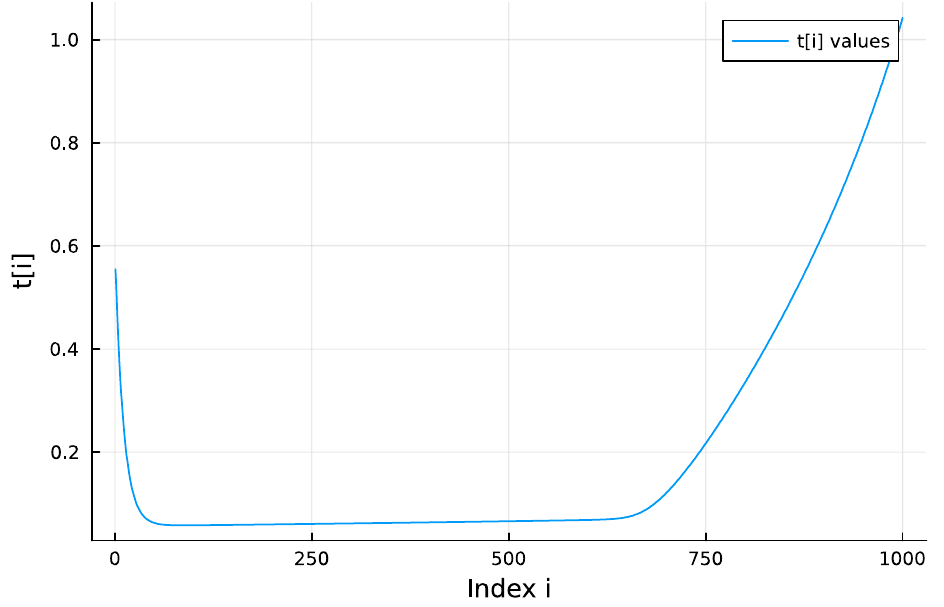}
        \caption{The values of $t_i$ of trajectory \textsc{EPS}$(\theta,t_1,\ldots,t_k)$ solving \tradeoff{2/3}. The corresponding trajectory is seen in Figure~\ref{fig: trajectoryl23}.
        }
        \label{fig: trajectoryl23-t}
    \end{subfigure}
    \hfill
    \begin{subfigure}[t]{0.32\textwidth}
        \centering
        \includegraphics[width=4.8cm]{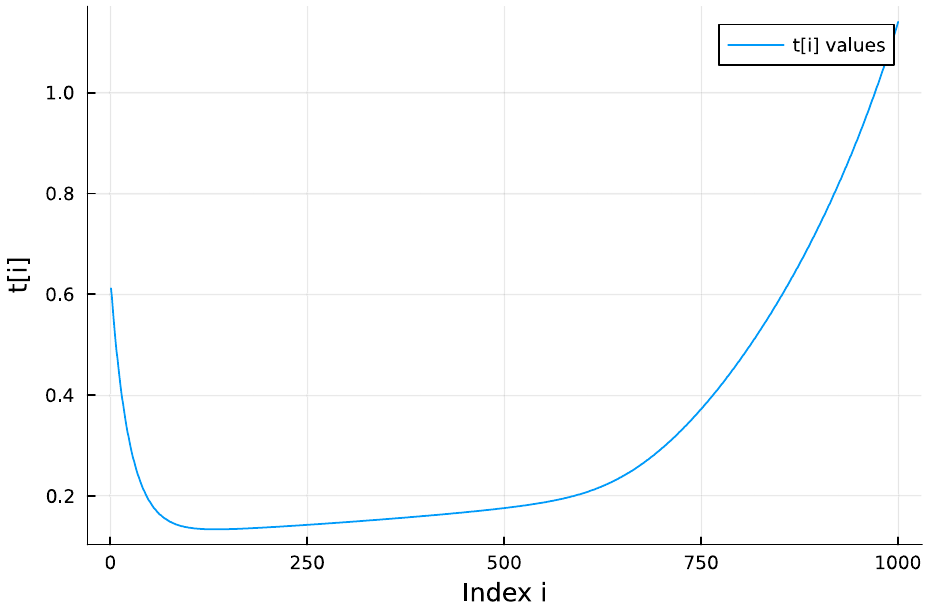}
        \caption{The values of $t_i$ of trajectory \textsc{EPS}$(\theta,t_1,\ldots,t_k)$ solving \tradeoff{1/3}. The corresponding trajectory is seen in Figure~\ref{fig: trajectoryl13}.
        }
        \label{fig: trajectoryl13-t}
    \end{subfigure}
    \caption{
   Agent's trajectories' parameter values $t_i$ for trade-off problem \tradeoff{\lambda} of the proof of Theorem~\ref{thm: general avg upper bound}, for $\lambda=1,2/3,1/3$. The parameter values corresponding to $\lambda=0$ (i.e. the best average performance of \textsc{EPS}$(\theta,t_1,\ldots,t_k)$ solving \shoreline{1}) can be seen in Figure~\ref{fig: tvaluesn1}.
   Notably, for $\lambda=1$, parameter values are nearly $0$, satisfying constraint $t_i \geq (1-\coss{\frac{c-2\theta}k})/\sinn{\frac{c-2\theta}k}$ tightly. For $\lambda=2/3$ the values are bounded away from the threshold. For $\lambda=1/3$, both the behavior of parameters $t_i$ as well as the trajectory resemble those corresponding to $\lambda=1$.
    }
    \label{fig: trajectoriesl-t}
\end{figure}

\end{document}